\documentclass[runningheads]{llncs}

\pdfoutput=1

\usepackage{amsmath}

\usepackage{xspace}
\usepackage{microtype} 

\usepackage[bookmarks, bookmarksnumbered, bookmarksopen
]{hyperref}
\hypersetup{
pdfauthor   = {Janis Voigtl\"{a}nder},
pdftitle    = {Free Theorems Simply, via Dinaturality},
pdfkeywords = {polymorphic types, free theorems, dinaturality}}

\usepackage
  {color}

\let\citep\cite

\makeatletter
\@ifundefined{lhs2tex.lhs2tex.sty.read}%
  {\@namedef{lhs2tex.lhs2tex.sty.read}{}%
   \newcommand\SkipToFmtEnd{}%
   \newcommand\EndFmtInput{}%
   \long\def\SkipToFmtEnd#1\EndFmtInput{}%
  }\SkipToFmtEnd

\newcommand\ReadOnlyOnce[1]{\@ifundefined{#1}{\@namedef{#1}{}}\SkipToFmtEnd}
\usepackage{amstext}
\usepackage{amssymb}
\usepackage{stmaryrd}
\DeclareFontFamily{OT1}{cmtex}{}
\DeclareFontShape{OT1}{cmtex}{m}{n}
  {<5><6><7><8>cmtex8
   <9>cmtex9
   <10><10.95><12><14.4><17.28><20.74><24.88>cmtex10}{}
\DeclareFontShape{OT1}{cmtex}{m}{it}
  {<-> ssub * cmtt/m/it}{}

\DeclareFontShape{OT1}{cmtt}{bx}{n}
  {<5><6><7><8>cmtt8
   <9>cmbtt9
   <10><10.95><12><14.4><17.28><20.74><24.88>cmbtt10}{}
\DeclareFontShape{OT1}{cmtex}{bx}{n}
  {<-> ssub * cmtt/bx/n}{}

\newcommand{\Conid}[1]{\mathit{#1}}
\newcommand{\Varid}[1]{\mathit{#1}}
\newcommand{\anonymous}{\kern0.06em \vbox{\hrule\@width.5em}}
\newcommand{\plus}{\mathbin{+\!\!\!+}}
\newcommand{\bind}{\mathbin{>\!\!\!>\mkern-6.7mu=}}

\usepackage{polytable}

\@ifundefined{mathindent}%
  {\newdimen\mathindent\mathindent\leftmargini}%
  {}%

\def\resethooks{%
  \global\let\SaveRestoreHook\empty
  \global\let\ColumnHook\empty}
\newcommand*{\savecolumns}[1][default]%
  {\g@addto@macro\SaveRestoreHook{\savecolumns[#1]}}
\newcommand*{\restorecolumns}[1][default]%
  {\g@addto@macro\SaveRestoreHook{\restorecolumns[#1]}}
\newcommand*{\aligncolumn}[2]%
  {\g@addto@macro\ColumnHook{\column{#1}{#2}}}

\resethooks

\newcommand{\onelinecommentchars}{\quad-{}- }
\newcommand{\commentbeginchars}{\enskip\{-}
\newcommand{\commentendchars}{-\}\enskip}

\newcommand{\visiblecomments}{%
  \let\onelinecomment=\onelinecommentchars
  \let\commentbegin=\commentbeginchars
  \let\commentend=\commentendchars}

\newcommand{\invisiblecomments}{%
  \let\onelinecomment=\empty
  \let\commentbegin=\empty
  \let\commentend=\empty}

\visiblecomments

\newlength{\blanklineskip}
\newcommand{\hsindent}[1]{\quad}
\let\hspre\empty
\let\hspost\empty

\EndFmtInput
\makeatother
\ReadOnlyOnce{forall.fmt}%
\makeatletter

\let\HaskellResetHook\empty
\newcommand*{\AtHaskellReset}[1]{%
  \g@addto@macro\HaskellResetHook{#1}}
\newcommand*{\HaskellReset}{\HaskellResetHook}

\newcommand\hsforall{\global\let\hsdot=\hsperiodonce}
\newcommand*\hsperiodonce[2]{#2\global\let\hsdot=\hscompose}
\newcommand*\hscompose[2]{#1}

\AtHaskellReset{\global\let\hsdot=\hscompose}

\HaskellReset

\makeatother
\EndFmtInput

\let\Conid\mathsf

\def\commentbegin{\quad$[\![\enskip$}
\def\commentend{$\enskip]\!]$}

\newcommand{\logrel}{\ensuremath{\Delta}}

\newcommand{\id}{\ensuremath{\mathit{id}}}

\begin{document}

\title{Free Theorems Simply, via Dinaturality}



\author{Janis Voigtl\"{a}nder}

\authorrunning{J.\ Voigtl\"{a}nder}

\institute{University of Duisburg-Essen, Germany\\
  \medskip
  \email{janis.voigtlaender@uni-due.de}
}

\maketitle

\begin{abstract}
Free theorems are a popular tool in reasoning about parametrically polymorphic code.
They are also of instructive use in teaching.
Their derivation, though, can be tedious, as it involves unfolding a lot of definitions, then hoping to be able to simplify the resulting logical formula to something nice and short.
Even in a mechanised generator it is not easy to get the right heuristics in place to achieve good outcomes.
Dinaturality is a categorical abstraction that captures many instances of free theorems.
Arguably, its origins are more conceptually involved to explain, though, and generating useful statements from it also has its pitfalls.
We present a simple approach for obtaining dinaturality-related free theorems from the standard formulation of relational parametricity in a rather direct
way.
It is conceptually appealing and easy to control and implement, as the provided Haskell code shows.
%
\end{abstract}

\section{Introduction}

Free theorems \citep{wad89} are an attractive means of reasoning about programs in a polymorphically typed language, predominantly used in a pure functional setting, but also available to functional-logic programmers \citep{mssv14}.
They have been employed for compiler optimisations \citep{glp93}
and other applications, and can also be used (when generated for deliberately arbitrary polymorphic types) to provide insight into the declarative nature of types and semantics of programs while teaching.
Free theorems are derived from relational parametricity \citep{rey83}, and the actual process of deriving them can be tedious.
We discuss an approach that side-steps the need to explicitly unfold definitions of relational actions and subsequently manipulate higher-order logic formulae.
That there is a relationship between relational parametricity and categorical dinaturality is not news at all \citep{bfss90}, and has been used to impressive effect lately \citep{hh15}, but we show that one can do without explicitly involving any category theory concepts, instead discovering all we need along the way.
Together with deterministic simplification rules, we obtain a compact and predictable free theorems generator.
We provide a neat implementation using the higher-order abstract syntax \citep{pe88} and normalisation by evaluation \citep{bs91} principles.

In the remainder of the paper, we are going to explain and discuss the standard approach of deriving free theorems via relational parametricity, first very informally (Section~\ref{sec:old}), then by somewhat superficially invoking its usual formal presentation (Section~\ref{sec:parametricity}), after which we ``discover'' our bridge to the simpler approach (Sections~\ref{sec:conjuring} and~\ref{sec:dinaturality}), and conclude with pragmatics and implementation (rest of Section~\ref{sec:new} and Section~\ref{sec:impl}).

\section{How free theorems are usually derived}\label{sec:old}

For the sake of simplicity, we consider only the case of declarations polymorphic in exactly one type variable, i.e., types like \ensuremath{(\alpha \to \Conid{Bool})\to [\mskip1.5mu \alpha \mskip1.5mu]\to \Conid{Maybe}\;\alpha } but not like \ensuremath{\alpha \to \beta \to (\alpha ,\beta )}.
Extension to cases like the latter would be possible.

\subsection{Constructing relations}\label{sec:relations}

The key to deriving free theorems is to interpret types as relations \citep{rey83,wad89}.
For example, given the type signature \ensuremath{\Varid{f}\mathbin{::}(\alpha \to \Conid{Bool})\to [\mskip1.5mu \alpha \mskip1.5mu]\to \Conid{Maybe}\;\alpha }, we replace the type variable \ensuremath{\alpha } by a relation variable \ensuremath{\mathcal{R}}, thus obtaining \ensuremath{(\mathcal{R}\to \Conid{Bool})\to ([\mskip1.5mu \mathcal{R}\mskip1.5mu]\to \Conid{Maybe}\;\mathcal{R})}.
Eventually, we will allow (nearly) arbitrary relations between closed types \ensuremath{\tau _{\mathrm{1}}} and \ensuremath{\tau _{\mathrm{2}}}, denoted \ensuremath{\mathcal{R}\in \Varid{Rel}(\tau _{\mathrm{1}},\tau _{\mathrm{2}})}, as interpretations for relation variables.
Also, there is a systematic way of reading expressions over relations as relations themselves.
In particular, 
\begin{itemize}
\item 
  base types like \ensuremath{\Conid{Bool}} and \ensuremath{\Conid{Int}} are read as identity relations, 

\item
  for relations \ensuremath{\mathcal{R}_{\mathrm{1}}} and \ensuremath{\mathcal{R}_{\mathrm{2}}}, we have
  \[\ensuremath{\mathcal{R}_{\mathrm{1}}\to \mathcal{R}_{\mathrm{2}}\mathrel{=}\{\mskip1.5mu (\Varid{f},\Varid{g})\mid \forall (\Varid{a},\Varid{b})\hsforall \in \mathcal{R}_{\mathrm{1}}\hsdot{\circ }{.}~(\Varid{f}\;\Varid{a},\Varid{g}\;\Varid{b})\in \mathcal{R}_{\mathrm{2}}\mskip1.5mu\}}\]
  and

\item
  every
  type constructor is read as an appropriate construction on relations;
  for example, the list type constructor maps every relation \ensuremath{\mathcal{R}\in \Varid{Rel}(\tau _{\mathrm{1}},\tau _{\mathrm{2}})} to the relation \ensuremath{[\mskip1.5mu \mathcal{R}\mskip1.5mu]\in \Varid{Rel}([\mskip1.5mu \tau _{\mathrm{1}}\mskip1.5mu],[\mskip1.5mu \tau _{\mathrm{2}}\mskip1.5mu])} defined by (the least fixpoint of)
  \[\ensuremath{[\mskip1.5mu \mathcal{R}\mskip1.5mu]\mathrel{=}\{\mskip1.5mu ([\mskip1.5mu \mskip1.5mu],[\mskip1.5mu \mskip1.5mu])\mskip1.5mu\}\mathop{\cup}\,\{\mskip1.5mu (\Varid{a}\mathbin{:}\Varid{as},\Varid{b}\mathbin{:}\Varid{bs})\mid (\Varid{a},\Varid{b})\in \mathcal{R},(\Varid{as},\Varid{bs})\in [\mskip1.5mu \mathcal{R}\mskip1.5mu]\mskip1.5mu\}}\]
  while the \ensuremath{\Conid{Maybe}} type constructor maps \ensuremath{\mathcal{R}\in \Varid{Rel}(\tau _{\mathrm{1}},\tau _{\mathrm{2}})} to \ensuremath{\Conid{Maybe}\;\mathcal{R}\in \Varid{Rel}(\Conid{Maybe}\;\tau _{\mathrm{1}},\Conid{Maybe}\;\tau _{\mathrm{2}})} defined by
  \[\ensuremath{\Conid{Maybe}\;\mathcal{R}\mathrel{=}\{\mskip1.5mu (\Conid{Nothing},\Conid{Nothing})\mskip1.5mu\}\mathop{\cup}\,\{\mskip1.5mu (\Conid{Just}\;\Varid{a},\Conid{Just}\;\Varid{b})\mid (\Varid{a},\Varid{b})\in \mathcal{R}\mskip1.5mu\}}\]
  and similarly for other datatypes.
\end{itemize}
The central statement of relational parametricity now is that
for every choice of \ensuremath{\tau _{\mathrm{1}}}, \ensuremath{\tau _{\mathrm{2}}}, and \ensuremath{\mathcal{R}}, the instantiations of the polymorphic \ensuremath{\Varid{f}} to types \ensuremath{\tau _{\mathrm{1}}} and \ensuremath{\tau _{\mathrm{2}}} are related by the relational interpretation of \ensuremath{\Varid{f}}'s type.
For the above example, this means that \ensuremath{(\Varid{f}_{\tau_1},\Varid{f}_{\tau_2})\in (\mathcal{R}\to \Varid{id}_\Conid{Bool})\to ([\mskip1.5mu \mathcal{R}\mskip1.5mu]\to \Conid{Maybe}\;\mathcal{R})}.
From now on, type subscripts will often be omitted since they can be easily inferred.

\subsection{Unfolding definitions}

To continue with the derivation of a free theorem in the standard way, one has to unfold the definitions of the various actions on relations described above.
For the example:
\begingroup\par\vskip\abovedisplayskip\noindent\advance\leftskip\mathindent\(
\begin{pboxed}\SaveRestoreHook
\column{B}{@{}>{\hspre}c<{\hspost}@{}}%
\column{BE}{@{}l@{}}%
\column{3}{@{}>{\hspre}l<{\hspost}@{}}%
\column{6}{@{}>{\hspre}l<{\hspost}@{}}%
\column{84}{@{}>{\hspre}l<{\hspost}@{}}%
\column{E}{@{}>{\hspre}l<{\hspost}@{}}%
\>[3]{}(\Varid{f},\Varid{f})\in (\mathcal{R}\to \Varid{id})\to ([\mskip1.5mu \mathcal{R}\mskip1.5mu]\to \Conid{Maybe}\;\mathcal{R}){}\<[E]%
\\
\>[B]{}\Leftrightarrow{}\<[BE]%
\>[6]{}\mbox{\commentbegin  definition of \ensuremath{\mathcal{R}_{\mathrm{1}}\to \mathcal{R}_{\mathrm{2}}}  \commentend}{}\<[E]%
\\
\>[B]{}\hsindent{3}{}\<[3]%
\>[3]{}\forall (\Varid{a},\Varid{b})\hsforall \in \mathcal{R}\to \Varid{id}\hsdot{\circ }{.}~(\Varid{f}\;\Varid{a},\Varid{f}\;\Varid{b})\in [\mskip1.5mu \mathcal{R}\mskip1.5mu]\to \Conid{Maybe}\;\mathcal{R}{}\<[E]%
\\
\>[B]{}\Leftrightarrow{}\<[BE]%
\>[6]{}\mbox{\commentbegin  again  \commentend}{}\<[E]%
\\
\>[B]{}\hsindent{3}{}\<[3]%
\>[3]{}\forall (\Varid{a},\Varid{b})\hsforall \in \mathcal{R}\to \Varid{id},(\Varid{c},\Varid{d})\in [\mskip1.5mu \mathcal{R}\mskip1.5mu]\hsdot{\circ }{.}~(\Varid{f}\;\Varid{a}\;\Varid{c},\Varid{f}\;\Varid{b}\;\Varid{d})\in \Conid{Maybe}\;\mathcal{R}{}\<[84]%
\>[84]{}~~~~~~\ensuremath{(\mathbin{*})}{}\<[E]%
\ColumnHook
\end{pboxed}
\)\par\vskip\belowdisplayskip\noindent\endgroup\resethooks
Now it is useful to specialise the relation \ensuremath{\mathcal{R}} to the ``graph'' of a function \ensuremath{\Varid{g}\mathbin{::}\tau _{\mathrm{1}}\to \tau _{\mathrm{2}}}, i.e., setting \ensuremath{\mathcal{R}\mathrel{=}\Varid{graph}(\Varid{g})\mathbin{:=}\{\mskip1.5mu (\Varid{x},\Varid{y})\mid \Varid{g}\;\Varid{x}\mathrel{=}\Varid{y}\mskip1.5mu\}\in \Varid{Rel}(\tau _{\mathrm{1}},\tau _{\mathrm{2}})}, and to realise that then \ensuremath{[\mskip1.5mu \mathcal{R}\mskip1.5mu]\mathrel{=}\Varid{graph}(\Varid{map}\;\Varid{g})} and \ensuremath{\Conid{Maybe}\;\mathcal{R}\mathrel{=}\Varid{graph}(\Varid{fmap}\;\Varid{g})}, so that we can continue as follows:
\begingroup\par\vskip\abovedisplayskip\noindent\advance\leftskip\mathindent\(
\begin{pboxed}\SaveRestoreHook
\column{B}{@{}>{\hspre}c<{\hspost}@{}}%
\column{BE}{@{}l@{}}%
\column{3}{@{}>{\hspre}l<{\hspost}@{}}%
\column{6}{@{}>{\hspre}l<{\hspost}@{}}%
\column{89}{@{}>{\hspre}l<{\hspost}@{}}%
\column{E}{@{}>{\hspre}l<{\hspost}@{}}%
\>[3]{}\forall (\Varid{a},\Varid{b})\hsforall \in \Varid{graph}(\Varid{g})\to \Varid{id},(\Varid{c},\Varid{d})\in \Varid{graph}(\Varid{map}\;\Varid{g})\hsdot{\circ }{.}~(\Varid{f}\;\Varid{a}\;\Varid{c},\Varid{f}\;\Varid{b}\;\Varid{d})\in \Varid{graph}(\Varid{fmap}\;\Varid{g}){}\<[E]%
\\
\>[B]{}\Leftrightarrow{}\<[BE]%
\>[6]{}\mbox{\commentbegin  \ensuremath{(\Varid{x},\Varid{y})\in \Varid{graph}(\Varid{h})} iff \ensuremath{\Varid{h}\;\Varid{x}\mathrel{=}\Varid{y}}  \commentend}{}\<[E]%
\\
\>[B]{}\hsindent{3}{}\<[3]%
\>[3]{}\forall (\Varid{a},\Varid{b})\hsforall \in \Varid{graph}(\Varid{g})\to \Varid{id},\Varid{c}\mathbin{::}[\mskip1.5mu \tau _{\mathrm{1}}\mskip1.5mu]\hsdot{\circ }{.}~\Varid{fmap}\;\Varid{g}\;(\Varid{f}\;\Varid{a}\;\Varid{c})\mathrel{=}\Varid{f}\;\Varid{b}\;(\Varid{map}\;\Varid{g}\;\Varid{c}){}\<[89]%
\>[89]{}~~~~~~\ensuremath{(\mathbin{**})}{}\<[E]%
\ColumnHook
\end{pboxed}
\)\par\vskip\belowdisplayskip\noindent\endgroup\resethooks
It remains to find out what \ensuremath{(\Varid{a},\Varid{b})\in \Varid{graph}(\Varid{g})\to \Varid{id}} means.
We can do so as follows:
\begingroup\par\vskip\abovedisplayskip\noindent\advance\leftskip\mathindent\(
\begin{pboxed}\SaveRestoreHook
\column{B}{@{}>{\hspre}c<{\hspost}@{}}%
\column{BE}{@{}l@{}}%
\column{3}{@{}>{\hspre}l<{\hspost}@{}}%
\column{6}{@{}>{\hspre}l<{\hspost}@{}}%
\column{E}{@{}>{\hspre}l<{\hspost}@{}}%
\>[3]{}(\Varid{a},\Varid{b})\in \Varid{graph}(\Varid{g})\to \Varid{id}{}\<[E]%
\\
\>[B]{}\Leftrightarrow{}\<[BE]%
\>[6]{}\mbox{\commentbegin  definition of \ensuremath{\mathcal{R}_{\mathrm{1}}\to \mathcal{R}_{\mathrm{2}}}  \commentend}{}\<[E]%
\\
\>[B]{}\hsindent{3}{}\<[3]%
\>[3]{}\forall (\Varid{x},\Varid{y})\hsforall \in \Varid{graph}(\Varid{g})\hsdot{\circ }{.}~(\Varid{a}\;\Varid{x},\Varid{b}\;\Varid{y})\in \Varid{id}{}\<[E]%
\\
\>[B]{}\Leftrightarrow{}\<[BE]%
\>[6]{}\mbox{\commentbegin  functions as relations  \commentend}{}\<[E]%
\\
\>[B]{}\hsindent{3}{}\<[3]%
\>[3]{}\forall \Varid{x}\hsforall \mathbin{::}\tau _{\mathrm{1}}\hsdot{\circ }{.}~\Varid{a}\;\Varid{x}\mathrel{=}\Varid{b}\;(\Varid{g}\;\Varid{x}){}\<[E]%
\\
\>[B]{}\Leftrightarrow{}\<[BE]%
\>[6]{}\mbox{\commentbegin  make pointfree  \commentend}{}\<[E]%
\\
\>[B]{}\hsindent{3}{}\<[3]%
\>[3]{}\Varid{a}\mathrel{=}\Varid{b}\hsdot{\circ }{.}\Varid{g}{}\<[E]%
\ColumnHook
\end{pboxed}
\)\par\vskip\belowdisplayskip\noindent\endgroup\resethooks
Finally, we obtain, for every \ensuremath{\Varid{f}\mathbin{::}(\alpha \to \Conid{Bool})\to [\mskip1.5mu \alpha \mskip1.5mu]\to \Conid{Maybe}\;\alpha }, \ensuremath{\Varid{g}\mathbin{::}\tau _{\mathrm{1}}\to \tau _{\mathrm{2}}}, \ensuremath{\Varid{b}\mathbin{::}\tau _{\mathrm{2}}\to \Conid{Bool}}, and \ensuremath{\Varid{c}\mathbin{::}[\mskip1.5mu \tau _{\mathrm{1}}\mskip1.5mu]},
\[\ensuremath{\Varid{fmap}\;\Varid{g}\;(\Varid{f}\;(\Varid{b}\hsdot{\circ }{.}\Varid{g})\;\Varid{c})\mathrel{=}\Varid{f}\;\Varid{b}\;(\Varid{map}\;\Varid{g}\;\Varid{c})}\]
or, if we prefer this
statement pointfree as well,
$\ensuremath{\Varid{fmap}\;\Varid{g}\hsdot{\circ }{.}\Varid{f}\;(\Varid{b}\hsdot{\circ }{.}\Varid{g})\mathrel{=}\Varid{f}\;\Varid{b}\hsdot{\circ }{.}\Varid{map}\;\Varid{g}}$.
The power of such statements is that \ensuremath{\Varid{f}} is only restricted by its type -- its behaviour can vary considerably within these confines, and still results obtained as free theorems will be guaranteed to hold.

\subsection{Typical complications}

So what is there not to like about the above procedure?
First of all, always unfolding the definitions of the relational actions -- specifically, the \ensuremath{\mathcal{R}_{\mathrm{1}}\to \mathcal{R}_{\mathrm{2}}} definition -- is tedious, though mechanical.
It typically brings us to something like \ensuremath{(\mathbin{*})} or \ensuremath{(\mathbin{**})} above.
Then, specifically if our \ensuremath{\Varid{f}} has a higher-order type, we will have to deal with preconditions like \ensuremath{(\Varid{a},\Varid{b})\in \mathcal{R}\to \Varid{id}} or \ensuremath{(\Varid{a},\Varid{b})\in \Varid{graph}(\Varid{g})\to \Varid{id}}.
Here we have seen, again by unfolding definitions, that the latter is equivalent to \ensuremath{\Varid{a}\mathrel{=}\Varid{b}\hsdot{\circ }{.}\Varid{g}}, which enabled simplification of statement \ensuremath{(\mathbin{**})} by eliminating the variable \ensuremath{\Varid{a}} completely. 
But in general this can become arbitrarily complicated.
If, for example, our \ensuremath{\Varid{f}} of interest had the type \ensuremath{(\alpha \to \alpha \to \Conid{Bool})\to [\mskip1.5mu \alpha \mskip1.5mu]\to [\mskip1.5mu \alpha \mskip1.5mu]}, we would have to deal with a precondition \ensuremath{(\Varid{a},\Varid{b})\in \Varid{graph}(\Varid{g})\to \Varid{graph}(\Varid{g})\to \Varid{id}} instead.
By similar steps as above, one can show that this is equivalent to
\ensuremath{\forall \Varid{x}\hsforall \mathbin{::}\tau _{\mathrm{1}},\Varid{y}\mathbin{::}\tau _{\mathrm{1}}\hsdot{\circ }{.}~\Varid{a}\;\Varid{x}\;\Varid{y}\mathrel{=}\Varid{b}\;(\Varid{g}\;\Varid{x})\;(\Varid{g}\;\Varid{y})}
or
\ensuremath{\forall \Varid{x}\hsforall \mathbin{::}\tau _{\mathrm{1}}\hsdot{\circ }{.}~\Varid{a}\;\Varid{x}\mathrel{=}\Varid{b}\;(\Varid{g}\;\Varid{x})\hsdot{\circ }{.}\Varid{g}}
or something even more cryptic if one insists on complete pointfreeness (to express the condition in the form ``\ensuremath{\Varid{a}\mathrel{=}\dots}'' in order to eliminate the explicit precondition by inlining).
One might prepare and keep in mind the simplifications of some common cases like those above, but in general, since the type of \ensuremath{\Varid{f}}, and thus of course also the types of higher-order arguments it may have, can be arbitrary and more ``exotic'' than above (in particular, possibly involving further nesting of function arrows -- consider, e.g., we had started with \ensuremath{\Varid{f}\mathbin{::}(([\mskip1.5mu \alpha \mskip1.5mu]\to \Conid{Int})\to \alpha )\to \alpha } as the target type), we are eventually down to unfolding the definitions of relational actions.
We can only \emph{hope} then to ultimately be able to also fold back into some compact form of precondition like was the case above.

Moreover, the picture is complicated by the fact that the procedure, exactly as described so far, applies only to the most basic language setting, namely a functional language in which there are no undefined values and all functions are total.
As soon as we consider a more realistic or interesting setting, some changes become required.
Typically that involves restricting the choice of relations over which one can quantify, but also changes to the relational actions that may or may not have a larger impact on the procedure of deriving free theorems.
Specifically, already when taking possible undefinedness and partiality of functions into account, one may only use relations that are strict (i.e., \ensuremath{(\bot ,\bot )\in \mathcal{R}}) and additionally has to use versions of datatype liftings that relate partial structures (e.g., \ensuremath{[\mskip1.5mu \mathcal{R}\mskip1.5mu]\mathrel{=}\{\mskip1.5mu (\bot ,\bot ),([\mskip1.5mu \mskip1.5mu],[\mskip1.5mu \mskip1.5mu])\mskip1.5mu\}\mathop{\cup}\,\{\mskip1.5mu (\Varid{a}\mathbin{:}\Varid{as},\Varid{b}\mathbin{:}\Varid{bs})\mid \dots\mskip1.5mu\}}).
This is not very severe yet, since strictness of relations simply translates into strictness of functions and connections like \ensuremath{[\mskip1.5mu \mathcal{R}\mskip1.5mu]\mathrel{=}\Varid{graph}(\Varid{map}\;\Varid{g})} for \ensuremath{\mathcal{R}\mathrel{=}\Varid{graph}(\Varid{g})} remain intact, so there is no considerable impact on the derivation steps.
But if one additionally takes Haskell's \ensuremath{\Varid{seq}}-primitive into account, more changes become required \citep{jv04}.
Now relations must also be total (i.e., \ensuremath{(\Varid{a},\Varid{b})\in \mathcal{R}} implies \ensuremath{\Varid{a}\mathrel{=}\bot \Leftrightarrow\Varid{b}\mathrel{=}\bot }) and additionally the relational action for function types must be changed to
\begingroup\par\vskip\abovedisplayskip\noindent\advance\leftskip\mathindent\(
\begin{pboxed}\SaveRestoreHook
\column{B}{@{}>{\hspre}l<{\hspost}@{}}%
\column{E}{@{}>{\hspre}l<{\hspost}@{}}%
\>[B]{}~~~~~~\mathcal{R}_{\mathrm{1}}\to \mathcal{R}_{\mathrm{2}}\mathrel{=}\{\mskip1.5mu (\Varid{f},\Varid{g})\mid \Varid{f}\mathrel{=}\bot \Leftrightarrow\Varid{g}\mathrel{=}\bot ,~\forall (\Varid{a},\Varid{b})\hsforall \in \mathcal{R}_{\mathrm{1}}\hsdot{\circ }{.}~(\Varid{f}\;\Varid{a},\Varid{g}\;\Varid{b})\in \mathcal{R}_{\mathrm{2}}\mskip1.5mu\}{}\<[E]%
\ColumnHook
\end{pboxed}
\)\par\vskip\belowdisplayskip\noindent\endgroup\resethooks
The latter \emph{does} have an impact on the derivation steps, since these typically (like in the examples above) use the definition of \ensuremath{\mathcal{R}_{\mathrm{1}}\to \mathcal{R}_{\mathrm{2}}} a lot, and now must manage the extra conditions concerning undefinedness.
Also, some simplifications become invalid in this setting.
Note that in the first example above we used that the precondition \ensuremath{\forall \Varid{x}\hsforall \mathbin{::}\tau _{\mathrm{1}}\hsdot{\circ }{.}~\Varid{a}\;\Varid{x}\mathrel{=}\Varid{b}\;(\Varid{g}\;\Varid{x})} is equivalent to \ensuremath{\Varid{a}\mathrel{=}\Varid{b}\hsdot{\circ }{.}\Varid{g}}.
But not in a language including \ensuremath{\Varid{seq}}, since in such a language eta-reduction is not generally valid (e.g., \ensuremath{\forall \Varid{x}\hsforall \hsdot{\circ }{.}~\bot \;\Varid{x}\mathrel{=}\bot \;(\Varid{id}\;\Varid{x})} but not \ensuremath{\bot \mathrel{=}\bot \hsdot{\circ }{.}\Varid{id}})!
We might still be safe, since the condition \ensuremath{\forall \Varid{x}\hsforall \mathbin{::}\tau _{\mathrm{1}}\hsdot{\circ }{.}~\Varid{a}\;\Varid{x}\mathrel{=}\Varid{b}\;(\Varid{g}\;\Varid{x})} is at least implied by \ensuremath{\Varid{a}\mathrel{=}\Varid{b}\hsdot{\circ }{.}\Varid{g}}, so depending on where that explicitly quantifying statement appeared in the overall statement we may obtain a weakening or a strengthening of that overall statement by replacing one condition by the other.
But such considerations require careful management of the preconditions and their positions in nested implication statements.
All this can still be done automatically~\citep{package-free-theorems}, but it is no pleasure.
There is not as much reuse as one might want, different simplification heuristics have to be used for different language settings, there is no really deterministic algorithm but instead some search involved, and sometimes the only ``simplification'' that seems to work is to unfold all definitions and leave it at that.
Moreover, if one were to move on and consider automatic generation of free theorems for further language settings, like imprecise error semantics \citep{sv09}, then the story would repeat itself.
There would be yet another set of changes to the basic definitions for relations and relational actions, new things to take care of during simplification of candidate free theorems, etc.

\subsection{Some problematic examples, and outlook at a remedy}\label{sec:problems}

Let us substantiate the above observations with some additional examples.
First we consider the declaration \ensuremath{\Varid{f}\mathbin{::}(([\mskip1.5mu \alpha \mskip1.5mu]\to \Conid{Int})\to \alpha )\to \alpha }.
Our existing free theorems generator library mentioned above~\citep{package-free-theorems} produces the statement that for every \ensuremath{\Varid{g}\mathbin{::}\tau _{\mathrm{1}}\to \tau _{\mathrm{2}}}, \ensuremath{\Varid{p}\mathbin{::}([\mskip1.5mu \tau _{\mathrm{1}}\mskip1.5mu]\to \Conid{Int})\to \tau _{\mathrm{1}}}, and \ensuremath{\Varid{q}\mathbin{::}([\mskip1.5mu \tau _{\mathrm{2}}\mskip1.5mu]\to \Conid{Int})\to \tau _{\mathrm{2}}}, it holds:\footnote{The new web UI for the library created by Joachim Breitner~\citep{nomeata-free-theorems} actually appears to not apply all possible simplifications, so the statement remains even a bit more complicated there.}
\begingroup\par\vskip\abovedisplayskip\noindent\advance\leftskip\mathindent\(
\begin{pboxed}\SaveRestoreHook
\column{B}{@{}>{\hspre}l<{\hspost}@{}}%
\column{3}{@{}>{\hspre}l<{\hspost}@{}}%
\column{E}{@{}>{\hspre}l<{\hspost}@{}}%
\>[B]{}(\forall \Varid{r}\hsforall \mathbin{::}[\mskip1.5mu \tau _{\mathrm{1}}\mskip1.5mu]\to \Conid{Int},\Varid{s}\mathbin{::}[\mskip1.5mu \tau _{\mathrm{2}}\mskip1.5mu]\to \Conid{Int}\hsdot{\circ }{.}~{}\<[E]%
\\
\>[B]{}\hsindent{3}{}\<[3]%
\>[3]{}(\forall \Varid{x}\hsforall \mathbin{::}[\mskip1.5mu \tau _{\mathrm{1}}\mskip1.5mu]\hsdot{\circ }{.}~\Varid{r}\;\Varid{x}\mathrel{=}\Varid{s}\;(\Varid{map}\;\Varid{g}\;\Varid{x}))\Rightarrow (\Varid{g}\;(\Varid{p}\;\Varid{r})\mathrel{=}\Varid{q}\;\Varid{s})){}\<[E]%
\\
\>[B]{}\Rightarrow (\Varid{g}\;(\Varid{f}\;\Varid{p})\mathrel{=}\Varid{f}\;\Varid{q}){}\<[E]%
\ColumnHook
\end{pboxed}
\)\par\vskip\belowdisplayskip\noindent\endgroup\resethooks
Arguably, it would have been more useful to be given the equivalent statement that for every \ensuremath{\Varid{f}}, \ensuremath{\Varid{g}}, \ensuremath{\Varid{p}} with types as above,
\begin{equation}
  \label{eq:1}
  \ensuremath{\Varid{g}\;(\Varid{f}\;\Varid{p})\mathrel{=}\Varid{f}\;(\lambda \Varid{s}\to \Varid{g}\;(\Varid{p}\;(\lambda \Varid{x}\to \Varid{s}\;(\Varid{map}\;\Varid{g}\;\Varid{x}))))}
\end{equation}
There is another free theorems generator as part of another tool, by Andrew Bromage~\citep{package-lambdabot}, and it does quite okay here, generating this:
\ensuremath{(\forall \Varid{p}\hsforall \hsdot{\circ }{.}~\Varid{g}\;(\Varid{h}\;(\Varid{p}\hsdot{\circ }{.}\Varid{map}\;\Varid{g}))\mathrel{=}\Varid{k}\;\Varid{p})\Rightarrow \Varid{g}\;(\Varid{f}\;\Varid{h})\mathrel{=}\Varid{f}\;\Varid{k}}.
But if we make the input type a bit more nasty by more nesting of function arrows, \ensuremath{\Varid{f}\mathbin{::}(((([\mskip1.5mu \alpha \mskip1.5mu]\to \Conid{Int})\to \Conid{Int})\to \Conid{Int})\to \alpha )\to \alpha }, then the existing generators differ only slightly from each other, and both yield something like the following:
\begingroup\par\vskip\abovedisplayskip\noindent\advance\leftskip\mathindent\(
\begin{pboxed}\SaveRestoreHook
\column{B}{@{}>{\hspre}l<{\hspost}@{}}%
\column{E}{@{}>{\hspre}l<{\hspost}@{}}%
\>[B]{}(\forall \Varid{r}\hsforall ,\Varid{s}\hsdot{\circ }{.}~(\forall \Varid{t}\hsforall ,\Varid{u}\hsdot{\circ }{.}~(\forall \Varid{w}\hsforall \hsdot{\circ }{.}~\Varid{t}\;(\Varid{w}\hsdot{\circ }{.}\Varid{map}\;\Varid{g})\mathrel{=}\Varid{u}\;\Varid{w})\Rightarrow (\Varid{r}\;\Varid{t}\mathrel{=}\Varid{s}\;\Varid{u}))\Rightarrow (\Varid{g}\;(\Varid{p}\;\Varid{r})\mathrel{=}\Varid{q}\;\Varid{s})){}\<[E]%
\\
\>[B]{}\Rightarrow (\Varid{g}\;(\Varid{f}\;\Varid{p})\mathrel{=}\Varid{f}\;\Varid{q}){}\<[E]%
\ColumnHook
\end{pboxed}
\)\par\vskip\belowdisplayskip\noindent\endgroup\resethooks
It would have been nicer to be given the following:
\begin{equation}
  \label{eq:2}
  \ensuremath{\Varid{g}\;(\Varid{f}\;\Varid{p})\mathrel{=}\Varid{f}\;(\lambda \Varid{s}\to \Varid{g}\;(\Varid{p}\;(\lambda \Varid{t}\to \Varid{s}\;(\lambda \Varid{w}\to \Varid{t}\;(\lambda \Varid{x}\to \Varid{w}\;(\Varid{map}\;\Varid{g}\;\Varid{x}))))))}
\end{equation}
which is exactly what the approach to be presented here will yield (modulo variable names).
Of course, one could invest into further post-processing steps in the existing generators to get from the scary form of the statement to the more readable, equivalent one.
But at some point, this will always be only partially successful.
Going from a compact relational expression to a quantifier-rich formula in higher-order logic through unfolding of definitions, and then trying to recover a more readable form via generic HOL formula manipulations, will generally be beaten by an approach better exploiting the structure present in the original type expression -- which is what we will do.
We will always generate a simple equation between two lambda-expressions, without precondition statements, as in (\ref{eq:1}) and (\ref{eq:2}) above.

Moreover, there is still the issue of the variability of free theorems between different language settings.
The generator inside Lambdabot does not consider such impact of language features, and thus the theorems it outputs are not safe in the presence of \ensuremath{\Varid{seq}}.
Our own previous generator does, and thus adds the proper extra conditions concerning undefinedness.
For example, for the more complicated of the two types considered above, the output then is (besides a strictness and totality condition imposed on \ensuremath{\Varid{g}})\footnote{Something we will not mention again and again is that \ensuremath{\Varid{g}} is also itself non-\ensuremath{\bot }. Disregarding types that contain only \ensuremath{\bot }, this follows from totality of \ensuremath{\Varid{g}} anyway.}:
\begingroup\par\vskip\abovedisplayskip\noindent\advance\leftskip\mathindent\(
\begin{pboxed}\SaveRestoreHook
\column{B}{@{}>{\hspre}l<{\hspost}@{}}%
\column{4}{@{}>{\hspre}l<{\hspost}@{}}%
\column{24}{@{}>{\hspre}l<{\hspost}@{}}%
\column{41}{@{}>{\hspre}l<{\hspost}@{}}%
\column{44}{@{}>{\hspre}l<{\hspost}@{}}%
\column{E}{@{}>{\hspre}l<{\hspost}@{}}%
\>[B]{}({}\<[4]%
\>[4]{}(\Varid{p}\neq\bot \Leftrightarrow\Varid{q}\neq\bot ){}\<[E]%
\\
\>[4]{}\mathrel{\wedge}(\forall \Varid{r}\hsforall ,\Varid{s}\hsdot{\circ }{.}~{}\<[24]%
\>[24]{}(\forall \Varid{t}\hsforall ,\Varid{u}\hsdot{\circ }{.}~{}\<[41]%
\>[41]{}({}\<[44]%
\>[44]{}(\Varid{t}\neq\bot \Leftrightarrow\Varid{u}\neq\bot ){}\<[E]%
\\
\>[44]{}\mathrel{\wedge}(\forall \Varid{v}\hsforall ,\Varid{w}\hsdot{\circ }{.}~(\forall \Varid{x}\hsforall \hsdot{\circ }{.}~\Varid{v}\;\Varid{x}\mathrel{=}\Varid{w}\;(\Varid{map}\;\Varid{g}\;\Varid{x}))\Rightarrow (\Varid{t}\;\Varid{v}\mathrel{=}\Varid{u}\;\Varid{w}))){}\<[E]%
\\
\>[41]{}\Rightarrow (\Varid{r}\;\Varid{t}\mathrel{=}\Varid{s}\;\Varid{u})){}\<[E]%
\\
\>[24]{}\Rightarrow (\Varid{g}\;(\Varid{p}\;\Varid{r})\mathrel{=}\Varid{q}\;\Varid{s}))){}\<[E]%
\\
\>[B]{}\Rightarrow (\Varid{g}\;(\Varid{f}\;\Varid{p})\mathrel{=}\Varid{f}\;\Varid{q}){}\<[E]%
\ColumnHook
\end{pboxed}
\)\par\vskip\belowdisplayskip\noindent\endgroup\resethooks
In contrast, with the approach to be presented we will get:
\begingroup\par\vskip\abovedisplayskip\noindent\advance\leftskip\mathindent\(
\begin{pboxed}\SaveRestoreHook
\column{B}{@{}>{\hspre}l<{\hspost}@{}}%
\column{E}{@{}>{\hspre}l<{\hspost}@{}}%
\>[B]{}~~~~~\Varid{g}\;(\Varid{f}\;(\lambda \Varid{s}\to \Varid{p}\;(\lambda \Varid{t}\to \Varid{s}\;(\lambda \Varid{w}\to \Varid{t}\;(\lambda \Varid{x}\to \Varid{w}\;\Varid{x}))))){}\<[E]%
\\
\>[B]{}~~~~~~~\mathrel{=}{}\<[E]%
\\
\>[B]{}~~~~~\Varid{f}\;(\lambda \Varid{s}\to \Varid{g}\;(\Varid{p}\;(\lambda \Varid{t}\to \Varid{s}\;(\lambda \Varid{w}\to \Varid{t}\;(\lambda \Varid{x}\to \Varid{w}\;(\Varid{map}\;\Varid{g}\;\Varid{x})))))){}\<[E]%
\ColumnHook
\end{pboxed}
\)\par\vskip\belowdisplayskip\noindent\endgroup\resethooks
which \dots
\begin{enumerate}
\item
  \dots is almost as strong as the more complicated formula above it.
  The only thing that makes it weaker is that it does not express that the corner cases \ensuremath{\Varid{g}\;(\Varid{f}\;\bot )\mathrel{=}\Varid{f}\;\bot } and \ensuremath{\Varid{g}\;(\Varid{f}\;\Varid{p'})\mathrel{=}\Varid{f}\;(\Varid{g}\hsdot{\circ }{.}\Varid{p'})} with \ensuremath{\Varid{p'}} any of \ensuremath{(\lambda \Varid{s}\to \Varid{p}\;\bot )}, \ensuremath{(\lambda \Varid{s}\to \Varid{p}\;(\lambda \Varid{t}\to \bot ))}, \ensuremath{(\lambda \Varid{s}\to \Varid{p}\;(\lambda \Varid{t}\to \Varid{s}\;\bot ))}, \dots, \ensuremath{(\lambda \Varid{s}\to \Varid{p}\;(\lambda \Varid{t}\to \Varid{s}\;(\lambda \Varid{w}\to \Varid{t}\;(\lambda \Varid{x}\to \Varid{w}\;\bot ))))} also hold.

\item
  \dots simply reduces to (\ref{eq:2}) in any functional language setting in which eta-reduction is valid.
  So we will not perform different derivations for different language settings.
  (Rather, eta-reduction, when applicable, can be applied as an afterthought -- which is exactly what our implementation will do.)
\end{enumerate}

To top the mentioned benefits, the approach to free theorems derivation we will discuss is much simpler than the previous one -- simpler both conceptually (and thus also when one wants to obtain free theorems by hand) as well as when implementing it.
In fact, the generator code takes up only half a page in the appendix (without counting the code for implementing the eta-reduction functionality) -- a small fraction of the size of the corresponding code in the existing free theorems generators.\footnote{Additional code for parsing input strings into type expressions and pretty-printing generated theorem expressions back into pleasingly looking strings is of comparable complexity between the different generators.}

There is one gotcha.
It is not \emph{always} possible to express a free theorem simply as an equation without preconditions.
A typical example is the type \ensuremath{\Varid{f}\mathbin{::}(\alpha \to \alpha )\to \alpha \to \alpha }.
Its general free theorem is:
\begingroup\par\vskip\abovedisplayskip\noindent\advance\leftskip\mathindent\(
\begin{pboxed}\SaveRestoreHook
\column{B}{@{}>{\hspre}l<{\hspost}@{}}%
\column{E}{@{}>{\hspre}l<{\hspost}@{}}%
\>[B]{}~~~~~(\Varid{g}\hsdot{\circ }{.}\Varid{h}\mathrel{=}\Varid{k}\hsdot{\circ }{.}\Varid{g})\Rightarrow (\Varid{g}\hsdot{\circ }{.}\Varid{f}\;\Varid{h}\mathrel{=}\Varid{f}\;\Varid{k}\hsdot{\circ }{.}\Varid{g}){}\<[E]%
\ColumnHook
\end{pboxed}
\)\par\vskip\belowdisplayskip\noindent\endgroup\resethooks
Since even for fixed \ensuremath{\Varid{g}}, neither of \ensuremath{\Varid{h}} and \ensuremath{\Varid{k}} uniquely determines the other here, the precondition \ensuremath{\Varid{g}\hsdot{\circ }{.}\Varid{h}\mathrel{=}\Varid{k}\hsdot{\circ }{.}\Varid{g}} cannot be avoided by some way of inlining or other strategy.
The dinaturality-based approach will instead generate the unconditional statement
\begingroup\par\vskip\abovedisplayskip\noindent\advance\leftskip\mathindent\(
\begin{pboxed}\SaveRestoreHook
\column{B}{@{}>{\hspre}l<{\hspost}@{}}%
\column{E}{@{}>{\hspre}l<{\hspost}@{}}%
\>[B]{}~~~~~\Varid{g}\;(\Varid{f}\;(\lambda \Varid{y}\to \Varid{p}\;(\Varid{g}\;\Varid{y}))\;\Varid{x})\mathrel{=}\Varid{f}\;(\lambda \Varid{y}\to \Varid{g}\;(\Varid{p}\;\Varid{y}))\;(\Varid{g}\;\Varid{x}){}\<[E]%
\ColumnHook
\end{pboxed}
\)\par\vskip\belowdisplayskip\noindent\endgroup\resethooks
i.e., setting \ensuremath{\Varid{h}} to \ensuremath{\Varid{p}\hsdot{\circ }{.}\Varid{g}} and \ensuremath{\Varid{k}} to \ensuremath{\Varid{g}\hsdot{\circ }{.}\Varid{p}} for some \ensuremath{\Varid{p}}, thus certainly satisfying \ensuremath{\Varid{g}\hsdot{\circ }{.}\Varid{h}\mathrel{=}\Varid{k}\hsdot{\circ }{.}\Varid{g}}, but losing some generality.
However, we believe we can say for what sort of types this will happen (see Section~\ref{sec:loss}).

\section{Free theorems simply, via dinaturality}\label{sec:new}

So, what is the magic sauce we are going to use?
We start from the simple observation that with the standard approach, once one has done the unfolding of definitions and subsequent simplifications/compactifications, one usually ends up with an equation (possibly with preconditions) between two expressions that look somewhat similar to each other.
For example, for type \ensuremath{\Varid{f}\mathbin{::}[\mskip1.5mu \alpha \mskip1.5mu]\to [\mskip1.5mu \alpha \mskip1.5mu]} one gets the equation \ensuremath{\Varid{map}\;\Varid{g}\;(\Varid{f}\;\Varid{xs})\mathrel{=}\Varid{f}\;(\Varid{map}\;\Varid{g}\;\Varid{xs})}, for type \ensuremath{\Varid{f}\mathbin{::}(\alpha \to \Conid{Bool})\to [\mskip1.5mu \alpha \mskip1.5mu]\to [\mskip1.5mu \alpha \mskip1.5mu]} one gets the equation \ensuremath{\Varid{map}\;\Varid{g}\;(\Varid{f}\;(\Varid{p}\hsdot{\circ }{.}\Varid{g})\;\Varid{xs})\mathrel{=}\Varid{f}\;\Varid{p}\;(\Varid{map}\;\Varid{g}\;\Varid{xs})}, etc.
There is certainly some regularity present: on one side \ensuremath{\Varid{map}\;\Varid{g}} happens ``before \ensuremath{\Varid{f}}'', on the other side it happens ``after \ensuremath{\Varid{f}}''; maybe \ensuremath{\Varid{g}} needs to be brought in at some other place in one or both of the two sides as well; but the expression structure is essentially the same on both sides.
In fact, given some experience with free theorems, one is often able to guess up front what the equation for a given type of \ensuremath{\Varid{f}} will look like.
But at present, to confirm it, one is still forced to do the chore of unfolding the definitions of the relational actions, then massaging the resulting formulae to hopefully bring them into the form one was expecting.
We will change that, by using what we call here the \emph{conjuring lemma of parametricity}.
It was previously stated for a functional-logic language to simplify derivation of free theorems in that setting \citep{mssv14} (Theorem~7.8 there), but will be used for (sublanguages of) Haskell here.
To justify it, we need a brief excursion (some readers may want to largely skip) into how relational parametricity
is usually formulated.

\subsection{Usual formulation of relational parametricity}\label{sec:parametricity}

Putting aside notational variations, as well as the fact that the exact form would differ a bit depending on whether one bases one's formalisation on a denotational or on an operational semantics (typically of a polymorphic lambda-calculus with some extensions, not full Haskell), one essentially always has the following theorem (sometimes called just the \emph{fundamental lemma of logical relations}).
Some explanations, such as what $\Delta$ stands for, are given below it.
\begin{theorem}[Relational Parametricity]\label{thm:standard}
  ~
  \begin{enumerate}
  \item 
    If \ensuremath{\Varid{e}} is a closed term (containing no free term variables, but also no free type variables) of a closed type \ensuremath{\tau }, then $(e, e) \ensuremath{\in } \logrel_{\emptyset,\ensuremath{\tau }}$.
  \item 
    If \ensuremath{\Varid{e}} is a closed term (in the sense of containing no free term variables) of a type polymorphic in one type variable, say \ensuremath{\sigma } containing free type variable \ensuremath{\alpha }, then for every choice of closed types \ensuremath{\tau _{\mathrm{1}}}, \ensuremath{\tau _{\mathrm{2}}}, and \ensuremath{\mathcal{R}\in \Varid{Rel}(\tau _{\mathrm{1}},\tau _{\mathrm{2}})}, we have $(e[\tau_1/\alpha], e[\tau_2/\alpha]) \ensuremath{\in } \logrel_{[\alpha\mapsto\ensuremath{\mathcal{R}}],\ensuremath{\sigma }}$.
  \item 
    If \ensuremath{\Varid{e}} is a polymorphic term as above, of type \ensuremath{\sigma } containing free type variable \ensuremath{\alpha }, but now possibly also containing a free term variable \ensuremath{\Varid{x}} of some type \ensuremath{\sigma '} possibly containing the free type variable \ensuremath{\alpha } as well, then for every choice of closed types \ensuremath{\tau _{\mathrm{1}}}, \ensuremath{\tau _{\mathrm{2}}}, and \ensuremath{\mathcal{R}} as above, and closed terms $e_1 :: \sigma'[\tau_1/\alpha]$ and $e_2 :: \sigma'[\tau_2/\alpha]$ such that $(e_1, e_2) \ensuremath{\in } \logrel_{[\alpha\mapsto\ensuremath{\mathcal{R}}], \ensuremath{\sigma '}}$, we have $(e[\tau_1/\alpha,e_1/x], e[\tau_2/\alpha,e_2/x]) \ensuremath{\in } \logrel_{[\alpha\mapsto\ensuremath{\mathcal{R}}],\ensuremath{\sigma }}$.
  \end{enumerate}
\end{theorem}
Now, the promised explanations:
\begin{itemize}
\item
  The notation $\logrel_{\rho,\ensuremath{\sigma }}$ corresponds to the construction of relations from types (as in Section~\ref{sec:relations}), where $\rho$ keeps track of the interpretation of any type variables by chosen relations.
  For example, $\logrel_{\emptyset,\ensuremath{\Conid{Int}\to [\mskip1.5mu \Conid{Bool}\mskip1.5mu]}}$ would be $\ensuremath{\Varid{id}}_{\ensuremath{\Conid{Int}}}\ensuremath{\to }[\ensuremath{\Varid{id}}_{\ensuremath{\Conid{Bool}}}]$ and $\logrel_{[\alpha\mapsto\ensuremath{\mathcal{R}}],\ensuremath{[\mskip1.5mu \alpha \mskip1.5mu]\to \alpha }}$ would be \ensuremath{[\mskip1.5mu \mathcal{R}\mskip1.5mu]\to \mathcal{R}}.

\item
  For any closed type \ensuremath{\tau }, the relation $\logrel_{\emptyset,\ensuremath{\tau }}$ (in fact, any $\logrel_{\rho,\ensuremath{\tau }}$) turns out to just be the identity relation at type \ensuremath{\tau }.
  As such, $(e, e) \ensuremath{\in } \logrel_{\emptyset,\ensuremath{\tau }}$ in the first item of the theorem may appear to state a triviality.
  However, if one explicitly handles abstraction and instantiation of type variables (we have not done so for the exposition in Section~\ref{sec:old}, because we anyway wanted to deal only with types polymorphic over exactly one type variable), then it is less so.
  One then introduces, alongside \ensuremath{\mathcal{R}_{\mathrm{1}}\to \mathcal{R}_{\mathrm{2}}} etc., a new relational action $\forall\mathcal{R}.~\mathcal{F}~\mathcal{R}$ (for mappings \ensuremath{\mathcal{F}} on relations), which is defined in exactly such a way that when moreover setting $\logrel_{\rho,\ensuremath{\forall \alpha \hsforall \hsdot{\circ }{.}\sigma }}=\forall\mathcal{R}.~\logrel_{\rho[\alpha\mapsto\ensuremath{\mathcal{R}}],\ensuremath{\sigma }}$, the statement $(e, e) \ensuremath{\in } \logrel_{\emptyset,\ensuremath{\forall \alpha \hsforall \hsdot{\circ }{.}\sigma }}$ reduces exactly to the statement in the second item of the theorem -- which then needs not to be explicitly made.
  The treatment is analogous if one has types polymorphic in more than one type variable, say \ensuremath{\tau \mathrel{=}\forall \alpha \hsforall \hsdot{\circ }{.}\forall \beta \hsforall \hsdot{\circ }{.}\sigma }, which explains how to deal with that case not considered in Section~\ref{sec:old}.

\item
  The choices of relations \ensuremath{\mathcal{R}\in \Varid{Rel}(\tau _{\mathrm{1}},\tau _{\mathrm{2}})} are not really completely arbitrary, instead depend on the language setting for which the parametricity theorem is stated and proved.
  As mentioned earlier, \ensuremath{\mathcal{R}} must be strict to take the presence of partial functions into account, and must be strict and total to take the presence of \ensuremath{\Varid{seq}} into account, and other restrictions may apply in other settings.

\item
  Even the third item of the theorem as stated above, adding the treatment of free term variables, is not yet the most general form.
  In general, the parametricity theorem is formulated for an arbitrary number of free type and term variables, in straightforward (but notationally tedious) extension of the formulations above.
  Just for the sake of exposition here, we have chosen the progression between the three items.
  Of course, usually not all three (or more/further ones) are shown, only one at the level of generality needed for a specific concern.
  In a short while, we will see that it can even be useful to consider the case where \ensuremath{\Varid{e}} does involve a type variable, and free term variables of types involving that type variable, but does itself \emph{not} have a polymorphic type.
\end{itemize}
Also, let us make explicit how Theorem~\ref{thm:standard} corresponds to the standard derivation approach for free theorems as described in Section~\ref{sec:old}.
Given a function \ensuremath{\Varid{f}} of type scheme \ensuremath{\sigma } polymorphic in \ensuremath{\alpha }, one would use the first or second item of the theorem to conclude $(f_{\tau_1}, f_{\tau_2}) \ensuremath{\in } \logrel_{[\alpha\mapsto\ensuremath{\mathcal{R}}],\ensuremath{\sigma }}$, then unfold the definition of $\logrel_{[\alpha\mapsto\ensuremath{\mathcal{R}}],\ensuremath{\sigma }}$, for example \ensuremath{(\Varid{f}_{\tau_1},\Varid{f}_{\tau_2})\in (\mathcal{R}\to \Varid{id}_\Conid{Bool})\to ([\mskip1.5mu \mathcal{R}\mskip1.5mu]\to \Conid{Maybe}\;\mathcal{R})} if \ensuremath{\sigma \mathrel{=}(\alpha \to \Conid{Bool})\to [\mskip1.5mu \alpha \mskip1.5mu]\to \Conid{Maybe}\;\alpha }, then continue from there, with all the tedious work this entails.

The trick now is to establish a lemma, actually a corollary, that does not even mention the relation construction $\Delta$, and that directly states an equality between expressions rather than something about relatedness.

\subsection{The conjuring lemma of parametricity}\label{sec:conjuring}

Before giving the lemma, let us give a brief example of the sort of term \ensuremath{\Varid{e}} that can appear in it, since without such an example it may be counterintuitive how \ensuremath{\Varid{e}} could ``involve \ensuremath{\alpha }'' but nevertheless have a closed overall type.
What this means is that \ensuremath{\Varid{e}} can be something like \ensuremath{\lambda \Varid{xs}\to \Varid{map}\;\Varid{post}\;(\Varid{f}\;(\Varid{map}\;\Varid{pre}\;\Varid{xs}))}.
In a context in which \ensuremath{\Varid{f}\mathbin{::}\forall \alpha \hsforall \hsdot{\circ }{.}[\mskip1.5mu \alpha \mskip1.5mu]\to [\mskip1.5mu \alpha \mskip1.5mu]} and \ensuremath{\Varid{pre}} and \ensuremath{\Varid{post}} are term variables typed \ensuremath{\tau _{\mathrm{1}}\to \alpha } and \ensuremath{\alpha \to \tau _{\mathrm{2}}} respectively, this \ensuremath{\Varid{e}} has the closed type \ensuremath{[\mskip1.5mu \tau _{\mathrm{1}}\mskip1.5mu]\to [\mskip1.5mu \tau _{\mathrm{2}}\mskip1.5mu]}, despite the fact that in order to write down \ensuremath{\Varid{e}} with explicit type annotations everywhere (i.e., on all subexpressions), one would also need to write down the type variable \ensuremath{\alpha } at some places.
Now the lemma, a corollary of the parametricity theorem.
\begin{lemma}[Conjuring Lemma]\label{lem:conjuringtrick:a}
~\\
Let \ensuremath{\tau }, \ensuremath{\tau _{\mathrm{1}}} and \ensuremath{\tau _{\mathrm{2}}} be closed types.
Let \ensuremath{\Varid{g}\mathbin{::}\tau _{\mathrm{1}}\to \tau _{\mathrm{2}}} be closed and:
\begin{itemize}
\item
  strict if we want to respect partially defined functions,
\item
  strict and total if we want to respect \ensuremath{\Varid{seq}}.
\end{itemize}
Let \ensuremath{\Varid{e}\mathbin{::}\tau } be a term possibly involving \ensuremath{\alpha } (but not in its own overall type, which is closed by assumption) and term variables \ensuremath{\Varid{pre}\mathbin{::}\tau _{\mathrm{1}}\to \alpha } and \ensuremath{\Varid{post}\mathbin{::}\alpha \to \tau _{\mathrm{2}}}, but no other free variables.
Then:
\[e[\tau_1/\alpha,\id_{\tau_1}/\ensuremath{\Varid{pre}},g/\ensuremath{\Varid{post}}] = e[\tau_2/\alpha,g/\ensuremath{\Varid{pre}},\id_{\tau_2}/\ensuremath{\Varid{post}}]\]
\end{lemma}
\begin{proof}
  The conditions on \ensuremath{\Varid{g}} (strictness/totality depending on language setting) guarantee that its graph can be used as an admissible \ensuremath{\mathcal{R}}.
  To apply the parametricity theorem (in its general form with arbitrarily many free variables), we need to establish $(\id_{\tau_1}, g) \ensuremath{\in } \logrel_{[\alpha\mapsto\ensuremath{\Varid{graph}(\Varid{g})}], \ensuremath{\tau _{\mathrm{1}}\to \alpha }}$ and $(g, \id_{\tau_2}) \ensuremath{\in } \logrel_{[\alpha\mapsto\ensuremath{\Varid{graph}(\Varid{g})}], \ensuremath{\alpha \to \tau _{\mathrm{2}}}}$.
  Since \ensuremath{\tau _{\mathrm{1}},\tau _{\mathrm{2}}} are closed types, these statements reduce to $(\id_{\tau_1}, g) \ensuremath{\in } \id_{\tau_1} \ensuremath{\to \Varid{graph}(\Varid{g})}$ and $(g, \id_{\tau_2}) \ensuremath{\in \Varid{graph}(\Varid{g})\to } \id_{\tau_2}$, respectively.
  Both of these hold in all the language settings considered (easy calculations; also note that \ensuremath{\Varid{g}\neq\bot } if \ensuremath{\Varid{g}} total), so the parametricity theorem lets us conclude 
  \[(e[\tau_1/\alpha,\id_{\tau_1}/\ensuremath{\Varid{pre}},g/\ensuremath{\Varid{post}}],e[\tau_2/\alpha,g/\ensuremath{\Varid{pre}},\id_{\tau_2}/\ensuremath{\Varid{post}}]) \ensuremath{\in } \logrel_{[\alpha\mapsto\ensuremath{\mathcal{R}}],\ensuremath{\tau }}\]
  from which the lemma's statement follows by $\logrel_{[\alpha\mapsto\ensuremath{\mathcal{R}}],\ensuremath{\tau }} = \id_\tau$ (recall: $\tau$ is closed).
\end{proof}
Let us reflect on what we have gained.
The conjuring lemma does not mention $\Delta$ from the previous subsection.
It holds in basically any language setting in which the (or better, a) parametricity theorem holds, no matter what the exact definitions of the relational actions (the unfolding steps employed for a concrete $\Delta$) are.
It is enough that a) the \emph{statement} of the parametricity theorem holds in the language setting under consideration, and that b) $(\id_{\tau_1}, g) \ensuremath{\in } \id_{\tau_1} \ensuremath{\to \Varid{graph}(\Varid{g})}$ and $(g, \id_{\tau_2}) \ensuremath{\in \Varid{graph}(\Varid{g})\to } \id_{\tau_2}$ do hold.
Both a) and b) are the case in all languages and $\Delta$-definitions we are aware of.
This does not just mean partiality and \ensuremath{\Varid{seq}} in Haskell, but also for example the setting with imprecise error semantics as studied in \citep{sv09}.
Even in work on parametricity and free theorems for a functional-logic language \citep{mssv14}, where the definition of $\Delta$, including the case \ensuremath{\mathcal{R}_{\mathrm{1}}\to \mathcal{R}_{\mathrm{2}}}, turns out somewhat differently (since having to deal with nondeterminism and thus with power domain types), the statement of the parametricity theorem and the definition of \ensuremath{\mathcal{R}_{\mathrm{1}}\to \mathcal{R}_{\mathrm{2}}} are such that the conjuring lemma holds.
Of course, whether the \ensuremath{\Varid{g}} in the lemma must be strict, or strict and total, or something else, does depend on the language setting, but this is not harmful, since it does not restrict us in our choice of \ensuremath{\Varid{e}}. 

Also, suppose the situation that some new datatype is to be considered.
Usually, this requires some new lifting to be defined and used for the relational interpretation of types.
Even though there is a standard recipe to follow, at least for run-of-the-mill algebraic datatypes, it is still work, and requires checking and of course building into a free theorems generator, along with appropriate simplification rules.
Not so if we use the conjuring lemma, which (while of course requiring an assertion that the parametricity theorem still holds even in the presence of the new datatype -- i.e., there must \emph{exist} an appropriate relational lifting) is not itself sensitive at all to how the new datatype is relationally interpreted.
If we can come up with interesting terms \ensuremath{\Varid{e}}, now possibly involving the new datatype, we are in good condition to prove new free theorems.

Before we consider the question whether we actually can, in general, come up with interesting terms \ensuremath{\Varid{e}}, let us do so for some specific examples.
We have already remarked, just before the conjuring lemma, that given \ensuremath{\Varid{f}\mathbin{::}\forall \alpha \hsforall \hsdot{\circ }{.}[\mskip1.5mu \alpha \mskip1.5mu]\to [\mskip1.5mu \alpha \mskip1.5mu]}, the term \ensuremath{\Varid{e}\mathrel{=}\lambda \Varid{xs}\to \Varid{map}\;\Varid{post}\;(\Varid{f}\;(\Varid{map}\;\Varid{pre}\;\Varid{xs}))} fits the bill, which means that the conjuring lemma gives us the following statement:
\begingroup\par\vskip\abovedisplayskip\noindent\advance\leftskip\mathindent\(
\begin{pboxed}\SaveRestoreHook
\column{B}{@{}>{\hspre}l<{\hspost}@{}}%
\column{E}{@{}>{\hspre}l<{\hspost}@{}}%
\>[B]{}~~~~~(\lambda \Varid{xs}\to \Varid{map}\;\Varid{g}\;(\Varid{f}\;(\Varid{map}\;\Varid{id}\;\Varid{xs})))\mathrel{=}(\lambda \Varid{xs}\to \Varid{map}\;\Varid{id}\;(\Varid{f}\;(\Varid{map}\;\Varid{g}\;\Varid{xs}))){}\<[E]%
\ColumnHook
\end{pboxed}
\)\par\vskip\belowdisplayskip\noindent\endgroup\resethooks
Using the additional knowledge that \ensuremath{\Varid{map}\;\Varid{id}\mathrel{=}\Varid{id}}, this is exactly the standard free theorem for said type of \ensuremath{\Varid{f}}, namely \ensuremath{\Varid{map}\;\Varid{g}\hsdot{\circ }{.}\Varid{f}\mathrel{=}\Varid{f}\hsdot{\circ }{.}\Varid{map}\;\Varid{g}}.

Let us try again, for the type \ensuremath{\Varid{f}\mathbin{::}\forall \alpha \hsforall \hsdot{\circ }{.}(\alpha \to \Conid{Bool})\to [\mskip1.5mu \alpha \mskip1.5mu]\to [\mskip1.5mu \alpha \mskip1.5mu]}.
We may ``know'' that we want \ensuremath{\Varid{map}\;\Varid{g}\;(\Varid{f}\;(\Varid{p}\hsdot{\circ }{.}\Varid{g})\;\Varid{xs})\mathrel{=}\Varid{f}\;\Varid{p}\;(\Varid{map}\;\Varid{g}\;\Varid{xs})}, but do not want to prove that statement with a lengthy derivation.
Imagining where \ensuremath{\Varid{pre}} and \ensuremath{\Varid{post}} should be put in order to make both sides of the desired statement an instance of a common term \ensuremath{\Varid{e}}, we may arrive at \ensuremath{\Varid{e}\mathrel{=}\lambda \Varid{p}\;\Varid{xs}\to \Varid{map}\;\Varid{post}\;(\Varid{f}\;(\Varid{p}\hsdot{\circ }{.}\Varid{post})\;(\Varid{map}\;\Varid{pre}\;\Varid{xs}))}, from which the conjuring lemma plus \ensuremath{\Varid{map}\;\Varid{id}\mathrel{=}\Varid{id}} rewriting gives us
\begingroup\par\vskip\abovedisplayskip\noindent\advance\leftskip\mathindent\(
\begin{pboxed}\SaveRestoreHook
\column{B}{@{}>{\hspre}l<{\hspost}@{}}%
\column{E}{@{}>{\hspre}l<{\hspost}@{}}%
\>[B]{}~~~~~(\lambda \Varid{p}\;\Varid{xs}\to \Varid{map}\;\Varid{g}\;(\Varid{f}\;(\Varid{p}\hsdot{\circ }{.}\Varid{g})\;\Varid{xs}))\mathrel{=}(\lambda \Varid{p}\;\Varid{xs}\to \Varid{f}\;(\Varid{p}\hsdot{\circ }{.}\Varid{id})\;(\Varid{map}\;\Varid{g}\;\Varid{xs})){}\<[E]%
\ColumnHook
\end{pboxed}
\)\par\vskip\belowdisplayskip\noindent\endgroup\resethooks
Should we also have rewritten \ensuremath{\Varid{p}\hsdot{\circ }{.}\Varid{id}} to \ensuremath{\Varid{p}}?
No, not in general!
In fact, \ensuremath{\Varid{p}\hsdot{\circ }{.}\Varid{id}\mathrel{=}\Varid{p}} is not valid in the presence of \ensuremath{\Varid{seq}}, and luckily there is no way to abuse the conjuring lemma for producing the not generally valid statement \ensuremath{\Varid{map}\;\Varid{g}\;(\Varid{f}\;(\Varid{p}\hsdot{\circ }{.}\Varid{g})\;\Varid{xs})\mathrel{=}\Varid{f}\;\Varid{p}\;(\Varid{map}\;\Varid{g}\;\Varid{xs})}.
Only after applying the lemma, when we commit to a specific language setting, we may decide that for us \ensuremath{\Varid{p}\hsdot{\circ }{.}\Varid{id}\mathrel{=}\Varid{p}} indeed holds.

To conclude this example exploration, let us consider the nasty type \ensuremath{\Varid{f}\mathbin{::}(((([\mskip1.5mu \alpha \mskip1.5mu]\to \Conid{Int})\to \Conid{Int})\to \Conid{Int})\to \alpha )\to \alpha } from Section~\ref{sec:problems}.
The choice \ensuremath{\Varid{e}\mathrel{=}\lambda \Varid{p}\to \Varid{post}\;(\Varid{f}\;(\lambda \Varid{s}\to \Varid{pre}\;(\Varid{p}\;(\lambda \Varid{t}\to \Varid{s}\;(\lambda \Varid{w}\to \Varid{t}\;(\lambda \Varid{x}\to \Varid{w}\;(\Varid{map}\;\Varid{pre}\;\Varid{x})))))))} gives us what we reported there as~(\ref{eq:2}).
We also remarked there that the approach presented here does not give us the various positive corner cases relevant in the presence of \ensuremath{\Varid{seq}}.
That is not fully true; actually the conjuring lemma gives us those as well, for example with \ensuremath{\Varid{e}\mathrel{=}\lambda \Varid{p}\to \Varid{post}\;(\Varid{f}\;(\lambda \Varid{s}\to \Varid{pre}\;(\Varid{p}\;(\lambda \Varid{t}\to \Varid{s}\;\bot ))))}, which is a valid input to Lemma~\ref{lem:conjuringtrick:a}, and gives us
\ensuremath{\Varid{g}\;(\Varid{f}\;(\lambda \Varid{s}\to \Varid{p}\;(\lambda \Varid{t}\to \Varid{s}\;\bot )))\mathrel{=}\Varid{f}\;(\lambda \Varid{s}\to \Varid{g}\;(\Varid{p}\;(\lambda \Varid{t}\to \Varid{s}\;\bot )))}.
But in what follows, we want to construct exactly one \ensuremath{\Varid{e}} for each type of \ensuremath{\Varid{f}}, and of course we opt for the supposedly most useful one, not for corner cases that ``just'' happen to also be valid.
So for said type of \ensuremath{\Varid{f}}, we want to, and will, construct the \ensuremath{\Varid{e}} which gives
\begingroup\par\vskip\abovedisplayskip\noindent\advance\leftskip\mathindent\(
\begin{pboxed}\SaveRestoreHook
\column{B}{@{}>{\hspre}l<{\hspost}@{}}%
\column{E}{@{}>{\hspre}l<{\hspost}@{}}%
\>[B]{}~~~~~\Varid{g}\;(\Varid{f}\;(\lambda \Varid{s}\to \Varid{p}\;(\lambda \Varid{t}\to \Varid{s}\;(\lambda \Varid{w}\to \Varid{t}\;(\lambda \Varid{x}\to \Varid{w}\;\Varid{x}))))){}\<[E]%
\\
\>[B]{}~~~~~~~\mathrel{=}{}\<[E]%
\\
\>[B]{}~~~~~\Varid{f}\;(\lambda \Varid{s}\to \Varid{g}\;(\Varid{p}\;(\lambda \Varid{t}\to \Varid{s}\;(\lambda \Varid{w}\to \Varid{t}\;(\lambda \Varid{x}\to \Varid{w}\;(\Varid{map}\;\Varid{g}\;\Varid{x})))))){}\<[E]%
\ColumnHook
\end{pboxed}
\)\par\vskip\belowdisplayskip\noindent\endgroup\resethooks
(or with left-hand side \ensuremath{\Varid{g}\;(\Varid{f}\;\Varid{p})}
in a world in which eta-reduction is valid).

\subsection{Constructing \ensuremath{\Varid{e}} -- discovering dinaturality}\label{sec:dinaturality}

Given some \ensuremath{\Varid{f}} of polymorphic type, we want to construct an \ensuremath{\Varid{e}} of closed type.
That seems easy, we could simply use \ensuremath{\Varid{e}\mathrel{=}\mathrm{42}}.
But no, of course we want \ensuremath{\Varid{e}} to use \ensuremath{\Varid{f}} in an interesting way.
In essence, we want it to ``touch'' each occurrence of the type variable \ensuremath{\alpha } in the type of \ensuremath{\Varid{f}}.
For doing so, \ensuremath{\Varid{e}} can use \ensuremath{\Varid{pre}\mathbin{::}\tau _{\mathrm{1}}\to \alpha } and \ensuremath{\Varid{post}\mathbin{::}\alpha \to \tau _{\mathrm{2}}}.
Some reflection shows that we should make a difference between positive and negative occurrences of \ensuremath{\alpha }, in the standard sense of polarity in function types.
That is, an occurrence of \ensuremath{\alpha } that is reached by an odd number of left-branching at function arrows (in the standard right-associative reading of \ensuremath{\to }) is considered a negative occurrence, others are considered positive occurrences.
So, for example, in the type \ensuremath{(\alpha \to \Conid{Bool})\to [\mskip1.5mu \alpha \mskip1.5mu]\to \Conid{Maybe}\;\alpha }, the first \ensuremath{\alpha } is positive, the second one is negative, and the third one is positive.
Then, we want to construct \ensuremath{\Varid{e}} such that negative occurrences of \ensuremath{\alpha } are replaced by \ensuremath{\tau _{\mathrm{1}}} and positive ones by \ensuremath{\tau _{\mathrm{2}}}.

This is doable by structural recursion on type expressions.
Specifically, the following function \ensuremath{\Varid{mono}_{\Varid{pre},\Varid{post}}(\sigma )} builds a term that maps an input of type \ensuremath{\sigma } to an output of a type with the same structure as \ensuremath{\sigma }, but made monomorphic according to the just described rule about negative and positive occurrences of \ensuremath{\alpha }.
So, for example, \ensuremath{\Varid{mono}_{\Varid{pre},\Varid{post}}((\alpha \to \Conid{Bool})\to [\mskip1.5mu \alpha \mskip1.5mu]\to \Conid{Maybe}\;\alpha )} maps an input of type \ensuremath{(\alpha \to \Conid{Bool})\to [\mskip1.5mu \alpha \mskip1.5mu]\to \Conid{Maybe}\;\alpha } to an output of type \ensuremath{(\tau _{\mathrm{2}}\to \Conid{Bool})\to [\mskip1.5mu \tau _{\mathrm{1}}\mskip1.5mu]\to \Conid{Maybe}\;\tau _{\mathrm{2}}}.
We do not prove the general behaviour, but it should be easy to see that \ensuremath{\Varid{mono}_{\Varid{pre},\Varid{post}}(\sigma )} does what we claim.
The defining equations we give should also be suggestive of what would have to be done if new datatypes (other than lists and \ensuremath{\Conid{Maybe}}) are introduced.
\begingroup\par\vskip\abovedisplayskip\noindent\advance\leftskip\mathindent\(
\begin{pboxed}\SaveRestoreHook
\column{B}{@{}>{\hspre}l<{\hspost}@{}}%
\column{3}{@{}>{\hspre}l<{\hspost}@{}}%
\column{40}{@{}>{\hspre}l<{\hspost}@{}}%
\column{E}{@{}>{\hspre}l<{\hspost}@{}}%
\>[3]{}\Varid{mono}_{\Varid{pre},\Varid{post}}(\alpha ){}\<[40]%
\>[40]{}\mathrel{=}\Varid{post}{}\<[E]%
\\
\>[3]{}\Varid{mono}_{\Varid{pre},\Varid{post}}(\Conid{Bool}){}\<[40]%
\>[40]{}\mathrel{=}\Varid{id}{}\<[E]%
\\
\>[3]{}\Varid{mono}_{\Varid{pre},\Varid{post}}(\Conid{Int}){}\<[40]%
\>[40]{}\mathrel{=}\Varid{id}{}\<[E]%
\\
\>[3]{}\Varid{mono}_{\Varid{pre},\Varid{post}}([\mskip1.5mu \sigma \mskip1.5mu]){}\<[40]%
\>[40]{}\mathrel{=}\Varid{map}\;\Varid{mono}_{\Varid{pre},\Varid{post}}(\sigma ){}\<[E]%
\\
\>[3]{}\Varid{mono}_{\Varid{pre},\Varid{post}}(\Conid{Maybe}\;\sigma ){}\<[40]%
\>[40]{}\mathrel{=}\Varid{fmap}\;\Varid{mono}_{\Varid{pre},\Varid{post}}(\sigma ){}\<[E]%
\\
\>[3]{}\Varid{mono}_{\Varid{pre},\Varid{post}}(\sigma _{\mathrm{1}}\to \sigma _{\mathrm{2}}){}\<[40]%
\>[40]{}\mathrel{=}\lambda \Varid{h}\to \Varid{mono}_{\Varid{pre},\Varid{post}}(\sigma _{\mathrm{2}})\hsdot{\circ }{.}\Varid{h}\hsdot{\circ }{.}\Varid{mono}_{\Varid{post},\Varid{pre}}(\sigma _{\mathrm{1}}){}\<[E]%
\ColumnHook
\end{pboxed}
\)\par\vskip\belowdisplayskip\noindent\endgroup\resethooks
Note the switching of \ensuremath{\Varid{pre}} and \ensuremath{\Varid{post}} in moving from \ensuremath{\sigma _{\mathrm{1}}\to \sigma _{\mathrm{2}}} to \ensuremath{\sigma _{\mathrm{1}}}.
Of course, in that last defining equation, the \ensuremath{\Varid{h}} must be a sufficiently fresh variable (also relative to \ensuremath{\Varid{pre}} and \ensuremath{\Varid{post}}).

Given \ensuremath{\Varid{f}} of polymorphic type \ensuremath{\forall \alpha \hsforall \hsdot{\circ }{.}\sigma }, we will be able to use the term \ensuremath{\Varid{e}\mathrel{=}\Varid{mono}_{\Varid{pre},\Varid{post}}(\sigma )\;\Varid{f}} in Lemma~\ref{lem:conjuringtrick:a}.
It is useful to notice then that (omitting explicit type instantiation and substitution):
\begin{itemize}
\item
  $\ensuremath{\Varid{mono}_{\Varid{pre},\Varid{post}}(\sigma )}[\id/\ensuremath{\Varid{pre}},g/\ensuremath{\Varid{post}}]=\ensuremath{\Varid{mono}_{\Varid{id},\Varid{g}}(\sigma )}$
\item
  $\ensuremath{\Varid{mono}_{\Varid{pre},\Varid{post}}(\sigma )}[g/\ensuremath{\Varid{pre}},\id/\ensuremath{\Varid{post}}]=\ensuremath{\Varid{mono}_{\Varid{g},\Varid{id}}(\sigma )}$
\end{itemize}
So our overall procedure now is to generate free theorems as follows:
\[\ensuremath{\Varid{mono}_{\Varid{id},\Varid{g}}(\sigma )\;\Varid{f}\mathrel{=}\Varid{mono}_{\Varid{g},\Varid{id}}(\sigma )\;\Varid{f}}\]
Category theory aficionados will recognise the concept of dinaturality here!

Let us try out the above for the example \ensuremath{\Varid{f}\mathbin{::}(\alpha \to \Conid{Bool})\to [\mskip1.5mu \alpha \mskip1.5mu]\to \Conid{Maybe}\;\alpha }.
We get:
\begingroup\par\vskip\abovedisplayskip\noindent\advance\leftskip\mathindent\(
\begin{pboxed}\SaveRestoreHook
\column{B}{@{}>{\hspre}c<{\hspost}@{}}%
\column{BE}{@{}l@{}}%
\column{4}{@{}>{\hspre}l<{\hspost}@{}}%
\column{12}{@{}>{\hspre}l<{\hspost}@{}}%
\column{20}{@{}>{\hspre}l<{\hspost}@{}}%
\column{67}{@{}>{\hspre}l<{\hspost}@{}}%
\column{E}{@{}>{\hspre}l<{\hspost}@{}}%
\>[4]{}\Varid{mono}_{\Varid{pre},\Varid{post}}((\alpha \to \Conid{Bool})\to [\mskip1.5mu \alpha \mskip1.5mu]\to \Conid{Maybe}\;\alpha ){}\<[E]%
\\
\>[B]{}\mathrel{=}{}\<[BE]%
\>[4]{}\lambda \Varid{h}_{1}\to \Varid{mono}_{\Varid{pre},\Varid{post}}([\mskip1.5mu \alpha \mskip1.5mu]\to \Conid{Maybe}\;\alpha )\hsdot{\circ }{.}\Varid{h}_{1}\hsdot{\circ }{.}\Varid{mono}_{\Varid{post},\Varid{pre}}(\alpha \to \Conid{Bool}){}\<[E]%
\\
\>[B]{}\mathrel{=}{}\<[BE]%
\>[4]{}\lambda \Varid{h}_{1}\to {}\<[12]%
\>[12]{}(\lambda \Varid{h}_{2}\to \Varid{mono}_{\Varid{pre},\Varid{post}}(\Conid{Maybe}\;\alpha )\hsdot{\circ }{.}\Varid{h}_{2}\hsdot{\circ }{.}\Varid{mono}_{\Varid{post},\Varid{pre}}([\mskip1.5mu \alpha \mskip1.5mu])){}\<[E]%
\\
\>[12]{}\hsdot{\circ }{.}\Varid{h}_{1}\hsdot{\circ }{.}{}\<[E]%
\\
\>[12]{}(\lambda \Varid{h}_{3}\to \Varid{mono}_{\Varid{post},\Varid{pre}}(\Conid{Bool})\hsdot{\circ }{.}\Varid{h}_{3}\hsdot{\circ }{.}\Varid{mono}_{\Varid{pre},\Varid{post}}(\alpha )){}\<[E]%
\\
\>[B]{}\mathrel{=}{}\<[BE]%
\>[4]{}\lambda \Varid{h}_{1}\to (\lambda \Varid{h}_{2}\to {}\<[20]%
\>[20]{}\Varid{fmap}\;\Varid{post}\hsdot{\circ }{.}\Varid{h}_{2}\hsdot{\circ }{.}\Varid{map}\;\Varid{pre})\hsdot{\circ }{.}\Varid{h}_{1}\hsdot{\circ }{.}(\lambda \Varid{h}_{3}\to {}\<[67]%
\>[67]{}\Varid{id}\hsdot{\circ }{.}\Varid{h}_{3}\hsdot{\circ }{.}\Varid{post}){}\<[E]%
\ColumnHook
\end{pboxed}
\)\par\vskip\belowdisplayskip\noindent\endgroup\resethooks
So the free theorem we get from this by instantiation is:
\begingroup\par\vskip\abovedisplayskip\noindent\advance\leftskip\mathindent\(
\begin{pboxed}\SaveRestoreHook
\column{B}{@{}>{\hspre}l<{\hspost}@{}}%
\column{20}{@{}>{\hspre}l<{\hspost}@{}}%
\column{62}{@{}>{\hspre}l<{\hspost}@{}}%
\column{E}{@{}>{\hspre}l<{\hspost}@{}}%
\>[B]{}~~~~~(\lambda \Varid{h}_{2}\to {}\<[20]%
\>[20]{}\Varid{fmap}\;\Varid{g}\hsdot{\circ }{.}\Varid{h}_{2}\hsdot{\circ }{.}\Varid{map}\;\Varid{id})\hsdot{\circ }{.}\Varid{f}\hsdot{\circ }{.}(\lambda \Varid{h}_{3}\to {}\<[62]%
\>[62]{}\Varid{id}\hsdot{\circ }{.}\Varid{h}_{3}\hsdot{\circ }{.}\Varid{g}){}\<[E]%
\\
\>[B]{}~~~~~~~\mathrel{=}{}\<[E]%
\\
\>[B]{}~~~~~(\lambda \Varid{h}_{2}\to {}\<[20]%
\>[20]{}\Varid{fmap}\;\Varid{id}\hsdot{\circ }{.}\Varid{h}_{2}\hsdot{\circ }{.}\Varid{map}\;\Varid{g})\hsdot{\circ }{.}\Varid{f}\hsdot{\circ }{.}(\lambda \Varid{h}_{3}\to {}\<[62]%
\>[62]{}\Varid{id}\hsdot{\circ }{.}\Varid{h}_{3}\hsdot{\circ }{.}\Varid{id}){}\<[E]%
\ColumnHook
\end{pboxed}
\)\par\vskip\belowdisplayskip\noindent\endgroup\resethooks
There are ample opportunities for further simplification here, but let us try to be systematic about this.

\subsection{Simplifying obtained statements}\label{sec:simpl}

The terms generated by \ensuremath{\Varid{mono}_{\Varid{pre},\Varid{post}}(\sigma )} contain a lot of function compositions, both outside and inside of \ensuremath{\Varid{map}}- and \ensuremath{\Varid{fmap}}-calls.
Moreover, many of the partners in those compositions will be \ensuremath{\Varid{id}}, either up front because of the cases with \ensuremath{\sigma } a base type, or later when \ensuremath{\Varid{pre}} or \ensuremath{\Varid{post}} is replaced by \ensuremath{\Varid{id}} (and the other by \ensuremath{\Varid{g}}).
So our primary strategy for simplification is to inline all the compositions, and while doing so eliminate all \ensuremath{\Varid{id}}-calls.
Additionally, all lambda-abstractions introduced by the \ensuremath{\Varid{mono}_{\Varid{pre},\Varid{post}}(\sigma _{\mathrm{1}}\to \sigma _{\mathrm{2}})} case will be provided with an argument and then beta-reduced.
There is no danger of term duplication here since the lambda-bound \ensuremath{\Varid{h}} is used only linearly in the right-hand side.

These considerations lead to the following syntactic simplification rules, to be applied to terms produced as \ensuremath{\Varid{mono}_{\Varid{id},\Varid{g}}(\sigma )\;\Varid{f}} or \ensuremath{\Varid{mono}_{\Varid{g},\Varid{id}}(\sigma )\;\Varid{f}}.
As usual, where new lambda-bound variables are introduced, they are assumed to be sufficiently fresh.
\begingroup\par\vskip\abovedisplayskip\noindent\advance\leftskip\mathindent\(
\begin{pboxed}\SaveRestoreHook
\column{B}{@{}>{\hspre}l<{\hspost}@{}}%
\column{3}{@{}>{\hspre}l<{\hspost}@{}}%
\column{28}{@{}>{\hspre}l<{\hspost}@{}}%
\column{E}{@{}>{\hspre}l<{\hspost}@{}}%
\>[3]{}\lfloor\Varid{id}~\Varid{t}\rfloor{}\<[28]%
\>[28]{}\mathrel{=}\Varid{t}{}\<[E]%
\\
\>[3]{}\lfloor\Varid{map}\;\Varid{f}~\Varid{t}\rfloor{}\<[28]%
\>[28]{}\mathrel{=}\Varid{map}\;(\lambda \Varid{v}\to \lfloor\Varid{f}~\Varid{v}\rfloor)\;\Varid{t}{}\<[E]%
\\
\>[3]{}\lfloor\Varid{fmap}\;\Varid{f}~\Varid{t}\rfloor{}\<[28]%
\>[28]{}\mathrel{=}\Varid{fmap}\;(\lambda \Varid{v}\to \lfloor\Varid{f}~\Varid{v}\rfloor)\;\Varid{t}{}\<[E]%
\\
\>[3]{}\lfloor(\lambda \Varid{h}\to \Varid{body})~\Varid{t}\rfloor{}\<[28]%
\>[28]{}\mathrel{=}\lambda \Varid{v}\to \lfloor\Varid{body}[t/h]~\Varid{v}\rfloor{}\<[E]%
\\
\>[3]{}\lfloor(\Varid{f}\hsdot{\circ }{.}\Varid{g})~\Varid{t}\rfloor{}\<[28]%
\>[28]{}\mathrel{=}\lfloor\Varid{f}~\lfloor\Varid{g}~\Varid{t}\rfloor\rfloor{}\<[E]%
\\
\>[3]{}\lfloor\Varid{f}~\Varid{t}\rfloor{}\<[28]%
\>[28]{}\mathrel{=}\Varid{f}\;\Varid{t}{}\<[E]%
\ColumnHook
\end{pboxed}
\)\par\vskip\belowdisplayskip\noindent\endgroup\resethooks
The last line is a catch-all case that is only used if none of the others apply.
In the case where the simplification function \ensuremath{\lfloor\cdot\rfloor} is applied to a term of the form \ensuremath{(\lambda \Varid{h}\to \Varid{body})\;\Varid{t}}, note that we are indeed entitled to eta-expand the beta-reduced version \ensuremath{\Varid{body}[t/h]} into \ensuremath{\lambda \Varid{v}\to \Varid{body}[t/h]\;\Varid{v}} (in order to subsequently apply \ensuremath{\lfloor\cdot\rfloor} recursively).
Said eta-expansion is type correct as well as semantically correct, since by analysing the \ensuremath{\Varid{mono}_{\Varid{pre},\Varid{post}}(\sigma )} function, which is the producer of the subexpression \ensuremath{(\lambda \Varid{h}\to \Varid{body})}, we know that \ensuremath{\Varid{body}}, and hence also \ensuremath{\Varid{body}[t/h]}, is a term formed by function composition, and since \ensuremath{\Varid{f}\hsdot{\circ }{.}\Varid{g}\mathrel{=}\lambda \Varid{v}\to (\Varid{f}\hsdot{\circ }{.}\Varid{g})\;\Varid{v}} is a valid equivalence even in language settings in which eta-reduction is not valid (and in which thus \ensuremath{\Varid{f}\mathrel{=}\lambda \Varid{v}\to \Varid{f}\;\Varid{v}} would not be in general okay).
The eta-expansions on the function arguments of \ensuremath{\Varid{map}}- and \ensuremath{\Varid{fmap}}-calls (again done to enable further simplification on \ensuremath{\Varid{f}\;\Varid{v}}) are also justified, since \ensuremath{\Varid{map}} and \ensuremath{\Varid{fmap}} use their function arguments only in specific, known ways:
\ensuremath{\Varid{map}\;\Varid{f}\;\Varid{t}} is indeed semantically equivalent to \ensuremath{\Varid{map}\;(\lambda \Varid{v}\to \Varid{f}\;\Varid{v})\;\Varid{t}}, since \ensuremath{\Varid{map}} does not use \ensuremath{\Varid{seq}}.
These considerations should convince us that \ensuremath{\lfloor\cdot\rfloor} transforms a term into a semantically equivalent one, hence is correct.
But is it also exhaustive, or can we accidentally skip transforming (and thus, simplifying) some part of the term produced by \ensuremath{\Varid{mono}_{\Varid{pre},\Varid{post}}(\sigma )}?
The best argument that we cannot, actually comes from the Haskell implementation given in the appendix, and will be discussed in Section~\ref{sec:impl}.

Let us be a bit more concrete again, and consider an example.
In the previous subsection, we generated \ensuremath{\Varid{mono}_{\Varid{pre},\Varid{post}}(\sigma )} for \ensuremath{\sigma \mathrel{=}(\alpha \to \Conid{Bool})\to [\mskip1.5mu \alpha \mskip1.5mu]\to \Conid{Maybe}\;\alpha }.
Let us now calculate \ensuremath{\lfloor\Varid{mono}_{\Varid{id},\Varid{g}}(\sigma )~\Varid{f}\rfloor} from this (of course, \ensuremath{\lfloor\Varid{mono}_{\Varid{g},\Varid{id}}(\sigma )~\Varid{f}\rfloor} would be very similar).
See Fig.~\ref{fig:simpl}.
The result is not yet fully satisfactory.
For one thing, \ensuremath{\lfloor\Varid{map}\;\Varid{id}~\Varid{v}_{2}\rfloor} was ``simplified'' to \ensuremath{\Varid{map}\;(\lambda \Varid{v}_{4}\to \Varid{v}_{4})\;\Varid{v}_{2}}.
Of course, we would prefer it to be simplified to just \ensuremath{\Varid{v}_{2}}.
This is easy to achieve by adding simplification rules like \ensuremath{\lfloor\Varid{map}\;\Varid{id}~\Varid{t}\rfloor\mathrel{=}\Varid{t}}.
In fact, our implementation does something more general, namely replacing the original first simplification rule \ensuremath{\lfloor\Varid{id}~\Varid{t}\rfloor\mathrel{=}\Varid{t}} by the following one: \ensuremath{\lfloor\Varid{f}~\Varid{t}\rfloor\mathrel{=}\Varid{t}} whenever \ensuremath{\Varid{f}} can be syntactically generated by the grammar \ensuremath{\Conid{Id}\mathrel{=}\Varid{id}\mid \Varid{map}\;\Conid{Id}\mid \Varid{fmap}\;\Conid{Id}\mid \Conid{Id}\hsdot{\circ }{.}\Conid{Id}}.
\begin{figure*}
$
\begin{array}{@{}l@{\;}l@{}}
  & \ensuremath{\lfloor(\lambda \Varid{h}_{1}\to (\lambda \Varid{h}_{2}\to \Varid{fmap}\;\Varid{g}\hsdot{\circ }{.}\Varid{h}_{2}\hsdot{\circ }{.}\Varid{map}\;\Varid{id})\hsdot{\circ }{.}\Varid{h}_{1}\hsdot{\circ }{.}(\lambda \Varid{h}_{3}\to \Varid{id}\hsdot{\circ }{.}\Varid{h}_{3}\hsdot{\circ }{.}\Varid{g}))~\Varid{f}\rfloor}\\
= & \ensuremath{\lambda \Varid{v}_{1}\to \lfloor((\lambda \Varid{h}_{2}\to \Varid{fmap}\;\Varid{g}\hsdot{\circ }{.}\Varid{h}_{2}\hsdot{\circ }{.}\Varid{map}\;\Varid{id})\hsdot{\circ }{.}\Varid{f}\hsdot{\circ }{.}(\lambda \Varid{h}_{3}\to \Varid{id}\hsdot{\circ }{.}\Varid{h}_{3}\hsdot{\circ }{.}\Varid{g}))~\Varid{v}_{1}\rfloor}\\
= & \ensuremath{\lambda \Varid{v}_{1}\to \lfloor(\lambda \Varid{h}_{2}\to \Varid{fmap}\;\Varid{g}\hsdot{\circ }{.}\Varid{h}_{2}\hsdot{\circ }{.}\Varid{map}\;\Varid{id})~\lfloor\Varid{f}~\lfloor(\lambda \Varid{h}_{3}\to \Varid{id}\hsdot{\circ }{.}\Varid{h}_{3}\hsdot{\circ }{.}\Varid{g})~\Varid{v}_{1}\rfloor\rfloor\rfloor}\\
= & \ensuremath{\lambda \Varid{v}_{1}\;\Varid{v}_{2}\to \lfloor(\Varid{fmap}\;\Varid{g}\hsdot{\circ }{.}\lfloor\Varid{f}~\lfloor(\lambda \Varid{h}_{3}\to \Varid{id}\hsdot{\circ }{.}\Varid{h}_{3}\hsdot{\circ }{.}\Varid{g})~\Varid{v}_{1}\rfloor\rfloor\hsdot{\circ }{.}\Varid{map}\;\Varid{id})~\Varid{v}_{2}\rfloor}\\
= & \ensuremath{\lambda \Varid{v}_{1}\;\Varid{v}_{2}\to \lfloor\Varid{fmap}\;\Varid{g}~\lfloor\lfloor\Varid{f}~\lfloor(\lambda \Varid{h}_{3}\to \Varid{id}\hsdot{\circ }{.}\Varid{h}_{3}\hsdot{\circ }{.}\Varid{g})~\Varid{v}_{1}\rfloor\rfloor~\lfloor\Varid{map}\;\Varid{id}~\Varid{v}_{2}\rfloor\rfloor\rfloor}\\
= & \ensuremath{\lambda \Varid{v}_{1}\;\Varid{v}_{2}\to \Varid{fmap}\;(\lambda \Varid{v}_{3}\to \lfloor\Varid{g}~\Varid{v}_{3}\rfloor)\;\lfloor\Varid{f}\;\lfloor(\lambda \Varid{h}_{3}\to \Varid{id}\hsdot{\circ }{.}\Varid{h}_{3}\hsdot{\circ }{.}\Varid{g})~\Varid{v}_{1}\rfloor~(\Varid{map}\;(\lambda \Varid{v}_{4}\to \lfloor\Varid{id}~\Varid{v}_{4}\rfloor)\;\Varid{v}_{2})\rfloor}\\
= & \ensuremath{\lambda \Varid{v}_{1}\;\Varid{v}_{2}\to \Varid{fmap}\;(\lambda \Varid{v}_{3}\to \Varid{g}\;\Varid{v}_{3})\;(\Varid{f}\;\lfloor(\lambda \Varid{h}_{3}\to \Varid{id}\hsdot{\circ }{.}\Varid{h}_{3}\hsdot{\circ }{.}\Varid{g})~\Varid{v}_{1}\rfloor\;(\Varid{map}\;(\lambda \Varid{v}_{4}\to \Varid{v}_{4})\;\Varid{v}_{2}))}\\
= & \ensuremath{\lambda \Varid{v}_{1}\;\Varid{v}_{2}\to \Varid{fmap}\;(\lambda \Varid{v}_{3}\to \Varid{g}\;\Varid{v}_{3})\;(\Varid{f}\;(\lambda \Varid{v}_{5}\to \lfloor(\Varid{id}\hsdot{\circ }{.}\Varid{v}_{1}\hsdot{\circ }{.}\Varid{g})~\Varid{v}_{5}\rfloor)\;(\Varid{map}\;(\lambda \Varid{v}_{4}\to \Varid{v}_{4})\;\Varid{v}_{2}))}\\
= & \ensuremath{\lambda \Varid{v}_{1}\;\Varid{v}_{2}\to \Varid{fmap}\;(\lambda \Varid{v}_{3}\to \Varid{g}\;\Varid{v}_{3})\;(\Varid{f}\;(\lambda \Varid{v}_{5}\to \lfloor\Varid{id}~\lfloor\Varid{v}_{1}~\lfloor\Varid{g}~\Varid{v}_{5}\rfloor\rfloor\rfloor)\;(\Varid{map}\;(\lambda \Varid{v}_{4}\to \Varid{v}_{4})\;\Varid{v}_{2}))}\\
= & \ensuremath{\lambda \Varid{v}_{1}\;\Varid{v}_{2}\to \Varid{fmap}\;(\lambda \Varid{v}_{3}\to \Varid{g}\;\Varid{v}_{3})\;(\Varid{f}\;(\lambda \Varid{v}_{5}\to \Varid{v}_{1}\;(\Varid{g}\;\Varid{v}_{5}))\;(\Varid{map}\;(\lambda \Varid{v}_{4}\to \Varid{v}_{4})\;\Varid{v}_{2}))}\\
\end{array}
$
\caption{An example calculation, for \ensuremath{\lfloor\Varid{mono}_{\Varid{id},\Varid{g}}((\alpha \to \Conid{Bool})\to [\mskip1.5mu \alpha \mskip1.5mu]\to \Conid{Maybe}\;\alpha )~\Varid{f}\rfloor}}
\label{fig:simpl}
\end{figure*}

Another issue is the unsatisfactory ``simplification'' of \ensuremath{\Varid{fmap}\;\Varid{g}} to \ensuremath{\Varid{fmap}\;(\lambda \Varid{v}_{3}\to \Varid{g}\;\Varid{v}_{3})} in Fig.~\ref{fig:simpl}.
To prevent it, but still keep the general rule \ensuremath{\lfloor\Varid{map}\;\Varid{f}~\Varid{t}\rfloor\mathrel{=}\Varid{map}\;(\lambda \Varid{v}\to \lfloor\Varid{f}~\Varid{v}\rfloor)\;\Varid{t}} in which the recursive descent can be important (say, if \ensuremath{\Varid{f}} is not just a function variable), we add the following simplification rule, which applies right after the one about generalised identities introduced above: \ensuremath{\lfloor\Varid{f}~\Varid{t}\rfloor\mathrel{=}\Varid{f}\;\Varid{t}} whenever \ensuremath{\Varid{f}} can be syntactically generated by the grammar \ensuremath{\Conid{Simple}\mathrel{=}\Varid{v}\mid \Varid{map}\;\Conid{Simple}\mid \Varid{fmap}\;\Conid{Simple}} (where \ensuremath{\Varid{v}} means any variable, including the \ensuremath{\Varid{g}} from \ensuremath{\Varid{mono}_{\Varid{id},\Varid{g}}(\sigma )}, say).
As a result, the calculation in Fig.~\ref{fig:simpl} would now yield the simplified term \ensuremath{\lambda \Varid{v}_{1}\;\Varid{v}_{2}\to \Varid{fmap}\;\Varid{g}\;(\Varid{f}\;(\lambda \Varid{v}_{3}\to \Varid{v}_{1}\;(\Varid{g}\;\Varid{v}_{3}))\;\Varid{v}_{2})}.

In summary, we generate the free theorem for a function \ensuremath{\Varid{f}} of polymorphic type \ensuremath{\sigma } as
\ensuremath{\lfloor\Varid{mono}_{\Varid{id},\Varid{g}}(\sigma )~\Varid{f}\rfloor\mathrel{=}\lfloor\Varid{mono}_{\Varid{g},\Varid{id}}(\sigma )~\Varid{f}\rfloor},
additionally using the simplification rules introduced and elaborated above.
There might still be eta-reducible expressions left in the produced terms.
But those are not necessarily safe to reduce in the presence of \ensuremath{\Varid{seq}}, so are left to a separate post-processing.

\subsection{About what generality is lost}\label{sec:loss}

There is one issue open from (the end of) Section~\ref{sec:problems}, where we promised to explain for what types the presented approach will not give an as general result as the existing generators, since it is impossible to express the general free theorem for those types as an equation without preconditions.
The criterion again depends on the notion of polarity in function types.
We believe it is an exact characterisation, but have no proof to show for it.

Let us annotate all parts of a type expression, not just the type variables, with their polarity.
So, for example, the type \ensuremath{(\alpha \to \alpha )\to \alpha \to \alpha } becomes $(\alpha^+ \ensuremath{\to } \alpha^-)^- \ensuremath{\to } (\alpha^- \ensuremath{\to } \alpha^+)^+$, the type \ensuremath{(\alpha \to \Conid{Bool})\to [\mskip1.5mu \alpha \mskip1.5mu]\to \Conid{Maybe}\;\alpha } becomes $(\alpha^+ \ensuremath{\to } \ensuremath{\Conid{Bool}}^-)^- \ensuremath{\to } ([\alpha^-]^- \ensuremath{\to } (\ensuremath{\Conid{Maybe}}~\alpha^+)^+)^+$, and the type \ensuremath{(\alpha \to \Conid{Bool})\to (\Conid{Bool}\to \alpha )\to [\mskip1.5mu \alpha \mskip1.5mu]\to \alpha } becomes $(\alpha^+ \ensuremath{\to } \ensuremath{\Conid{Bool}}^-)^- \ensuremath{\to } ((\ensuremath{\Conid{Bool}}^+ \ensuremath{\to } \alpha^-)^- \ensuremath{\to } ([\alpha^-]^- \ensuremath{\to } \alpha^+)^+)^+$.
The types for which stating a simple equation is not the most general free theorem are those which contain a negative subexpression in which both a positive and a negative \ensuremath{\alpha } appear.
For example, this is the case for $(\alpha^+ \ensuremath{\to } \alpha^-)^- \ensuremath{\to } (\alpha^- \ensuremath{\to } \alpha^+)^+$, but not for the other two types considered above.

\section{Implementation}\label{sec:impl}

Figs.~\ref{fig:generate}--\ref{fig:main}
in the appendix give the Haskell code for deriving free theorems using the presented approach, minus the code for lexing, parsing, and pretty-printing.
The actual generation work happens
in Fig.~\ref{fig:generate}: \ensuremath{\Varid{mono}} implements \ensuremath{\Varid{mono}_{\Varid{pre},\Varid{post}}(\sigma )}, \ensuremath{\Varid{apply}} implements \ensuremath{\lfloor\Varid{f}~\Varid{t}\rfloor}.
An
eta-reducer
is implemented in Fig.~\ref{fig:syntax}.
To encode lambda-terms (and substitution), we use higher-order abstract syntax, and normalisation by evaluation principles come into play as well.

The generator code is available online \citep{package-ft-generator}, and can be installed using standard Haskell package management tools.
Alternatively, the appendix contains some more step-by-step installation instructions.
To see the generator in action, see Fig.~\ref{fig:session} in the appendix.

From Section~\ref{sec:simpl} we still owe an argument that \ensuremath{\lfloor\cdot\rfloor} is exhaustive, i.e., that we cannot accidentally skip transforming (and thus, simplifying) some part of the term produced by \ensuremath{\Varid{mono}}.
So, consider the following.
In the implementation, \ensuremath{\lfloor\cdot\rfloor} does not take one argument term, but instead two, and they are differently typed.
Specifically, there is one syntax type for terms generated by \ensuremath{\Varid{mono}} and another syntax type for final output terms.
The type of \ensuremath{\lfloor\cdot\rfloor} is such that it always takes an \ensuremath{\Varid{f}} in the former syntax and a \ensuremath{\Varid{t}} in the latter syntax, in the form \ensuremath{\lfloor\Varid{f}~\Varid{t}\rfloor}.
The type of \ensuremath{\Varid{mono}} guarantees that it indeed generates \ensuremath{\Varid{f}} in the former syntax type (in essence, hence, this syntax type characterises a subclass of lambda-terms in which all terms possibly generated by \ensuremath{\Varid{mono}} live).
Since it is easy to see by inspection of the implementation that \ensuremath{\lfloor\cdot\rfloor} does an exhaustive case distinction on all possible forms of its first argument (it handles all constructor cases of its syntax type), and since \ensuremath{\lfloor\cdot\rfloor}'s output type is the one of final output terms, not of \ensuremath{\Varid{mono}}-generated terms, we know that no parts of the \ensuremath{\Varid{mono}_{\Varid{pre},\Varid{post}}(\sigma )}-term survive untouched.
In particular, the syntax and function types in the implementation also tell us that the catch-all case \ensuremath{\lfloor\Varid{f}~\Varid{t}\rfloor\mathrel{=}\Varid{f}\;\Varid{t}} will not have to deal with any \ensuremath{\Varid{f}} that still contains \ensuremath{\Varid{mono}}-material.
Instead, we know that if the catch-all case is reached, \ensuremath{\Varid{f}} is a variable or an already simplified final output term (which could have come into place via the substitution in an earlier recursive call \ensuremath{\lfloor\Varid{body}[t/h]~\Varid{v}\rfloor}).

\subsubsection*{Acknowledgements.}

Many of the ideas here, most notably the conjuring lemma (including that name), but also the criterion proposed in Section~\ref{sec:loss}, originated during past collaboration with Stefan Mehner.

\bibliographystyle{splncs04}
\bibliography{references}

\newpage

\appendix
\section{Implementation}

See Section~\ref{sec:impl} for some explanation of the implementation.
To run the generator locally, install the Haskell platform, download \url{https://hackage.haskell.org/package/ft-generator-1.0/ft-generator-1.0.tar.gz}, unpack it and navigate into the created directory, call \text{\ttfamily cabal~install}, and then the executable (the \text{\ttfamily cabal} command should tell where it has put it).

\begin{figure*}
\begingroup\par\vskip\abovedisplayskip\noindent\advance\leftskip\mathindent\(
\begin{pboxed}\SaveRestoreHook
\column{B}{@{}>{\hspre}l<{\hspost}@{}}%
\column{15}{@{}>{\hspre}l<{\hspost}@{}}%
\column{18}{@{}>{\hspre}l<{\hspost}@{}}%
\column{20}{@{}>{\hspre}l<{\hspost}@{}}%
\column{44}{@{}>{\hspre}l<{\hspost}@{}}%
\column{48}{@{}>{\hspre}l<{\hspost}@{}}%
\column{50}{@{}>{\hspre}l<{\hspost}@{}}%
\column{E}{@{}>{\hspre}l<{\hspost}@{}}%
\>[B]{}\mathbf{module}\;\Conid{Generate}\;(\Varid{mono},\Varid{apply})\;\mathbf{where}{}\<[E]%
\\[0.7ex]
\>[B]{}\mathbf{import}\;\Conid{Syntax}\;(\Conid{Type}\;(\mathinner{\ldotp\ldotp}),\Conid{Func}\;(\mathinner{\ldotp\ldotp}),\Conid{Term}\;(\mathinner{\ldotp\ldotp})){}\<[E]%
\\[0.7ex]
\>[B]{}\Varid{mono}\mathbin{::}\Conid{Type}\to \Conid{Func}\;\Varid{v}\to \Conid{Func}\;\Varid{v}\to \Conid{Func}\;\Varid{v}{}\<[E]%
\\
\>[B]{}\Varid{mono}\;\Conid{Alpha}\;{}\<[18]%
\>[18]{}\Varid{pre}\;\Varid{post}\mathrel{=}\Varid{post}{}\<[E]%
\\
\>[B]{}\Varid{mono}\;\Conid{Bool}\;{}\<[18]%
\>[18]{}\Varid{pre}\;\Varid{post}\mathrel{=}\Conid{Id}{}\<[E]%
\\
\>[B]{}\Varid{mono}\;\Conid{Int}\;{}\<[18]%
\>[18]{}\Varid{pre}\;\Varid{post}\mathrel{=}\Conid{Id}{}\<[E]%
\\
\>[B]{}\Varid{mono}\;(\Conid{List}\;\Varid{t})\;{}\<[18]%
\>[18]{}\Varid{pre}\;\Varid{post}\mathrel{=}\Conid{Map}\;\text{\ttfamily \char34 map\char34}\;(\Varid{mono}\;\Varid{t}\;\Varid{pre}\;\Varid{post}){}\<[E]%
\\
\>[B]{}\Varid{mono}\;(\Conid{Maybe}\;\Varid{t})\;{}\<[18]%
\>[18]{}\Varid{pre}\;\Varid{post}\mathrel{=}\Conid{Map}\;\text{\ttfamily \char34 fmap\char34}\;(\Varid{mono}\;\Varid{t}\;\Varid{pre}\;\Varid{post}){}\<[E]%
\\
\>[B]{}\Varid{mono}\;(\Varid{s}\mathbin{`\Conid{To}`}\Varid{t})\;{}\<[18]%
\>[18]{}\Varid{pre}\;\Varid{post}\mathrel{=}\Conid{Lambda}\;(\lambda \Varid{h}\to {}\<[44]%
\>[44]{}\Varid{mono}\;\Varid{t}\;\Varid{pre}\;\Varid{post}{}\<[E]%
\\
\>[44]{}\mathbin{`\Conid{Comp}`}{}\<[E]%
\\
\>[44]{}\Conid{Embed}\;\Varid{h}{}\<[E]%
\\
\>[44]{}\mathbin{`\Conid{Comp}`}{}\<[E]%
\\
\>[44]{}\Varid{mono}\;\Varid{s}\;\Varid{post}\;\Varid{pre}){}\<[E]%
\\[0.7ex]
\>[B]{}\Varid{apply}\mathbin{::}\Conid{Func}\;(\Conid{Term}\;\Varid{v})\to \Conid{Term}\;\Varid{v}\to \Conid{Term}\;\Varid{v}{}\<[E]%
\\
\>[B]{}\Varid{f}{}\<[15]%
\>[15]{}\mathbin{`\Varid{apply}`}\Varid{t}\mid \Varid{isId}\;\Varid{f}{}\<[50]%
\>[50]{}\mathrel{=}\Varid{t}{}\<[E]%
\\
\>[B]{}\Varid{f}{}\<[15]%
\>[15]{}\mathbin{`\Varid{apply}`}\Varid{t}\mid \Conid{Just}\;\Varid{f'}\leftarrow \Varid{isSimple}\;\Varid{f}{}\<[50]%
\>[50]{}\mathrel{=}\Varid{f'}\mathbin{`\Conid{Apply}`}\Varid{t}{}\<[E]%
\\
\>[B]{}\Conid{Map}\;\Varid{name}\;\Varid{f}{}\<[15]%
\>[15]{}\mathbin{`\Varid{apply}`}\Varid{t}\mathrel{=}\Conid{Const}\;\Varid{name}\mathbin{`\Conid{Apply}`}\Conid{Lambda'}\;(\lambda \Varid{v}\to \Varid{f}\mathbin{`\Varid{apply}`}(\Conid{Var}\;\Varid{v}))\mathbin{`\Conid{Apply}`}\Varid{t}{}\<[E]%
\\
\>[B]{}\Conid{Lambda}\;\Varid{f}{}\<[15]%
\>[15]{}\mathbin{`\Varid{apply}`}\Varid{t}\mathrel{=}\Conid{Lambda'}\;(\lambda \Varid{v}\to \Varid{f}\;\Varid{t}\mathbin{`\Varid{apply}`}(\Conid{Var}\;\Varid{v})){}\<[E]%
\\
\>[B]{}(\Varid{f}\mathbin{`\Conid{Comp}`}\Varid{g}){}\<[15]%
\>[15]{}\mathbin{`\Varid{apply}`}\Varid{t}\mathrel{=}\Varid{f}\mathbin{`\Varid{apply}`}(\Varid{g}\mathbin{`\Varid{apply}`}\Varid{t}){}\<[E]%
\\
\>[B]{}\Conid{Embed}\;\Varid{f}{}\<[15]%
\>[15]{}\mathbin{`\Varid{apply}`}\Varid{t}\mathrel{=}\Varid{f}\mathbin{`\Conid{Apply}`}\Varid{t}{}\<[E]%
\\[0.7ex]
\>[B]{}\Varid{isId}\mathbin{::}\Conid{Func}\;\Varid{v}\to \Conid{Bool}{}\<[E]%
\\
\>[B]{}\Varid{isId}\;\Conid{Id}{}\<[20]%
\>[20]{}\mathrel{=}\Conid{True}{}\<[E]%
\\
\>[B]{}\Varid{isId}\;(\Conid{Map}\;\anonymous \;\Varid{f}){}\<[20]%
\>[20]{}\mathrel{=}\Varid{isId}\;\Varid{f}{}\<[E]%
\\
\>[B]{}\Varid{isId}\;(\Varid{f}\mathbin{`\Conid{Comp}`}\Varid{g}){}\<[20]%
\>[20]{}\mathrel{=}\Varid{isId}\;\Varid{f}\mathop{\,\&\&\,}\Varid{isId}\;\Varid{g}{}\<[E]%
\\
\>[B]{}\Varid{isId}\;\anonymous {}\<[20]%
\>[20]{}\mathrel{=}\Conid{False}{}\<[E]%
\\[0.7ex]
\>[B]{}\Varid{isSimple}\mathbin{::}\Conid{Func}\;(\Conid{Term}\;\Varid{v})\to \Conid{Maybe}\;(\Conid{Term}\;\Varid{v}){}\<[E]%
\\
\>[B]{}\Varid{isSimple}\;(\Conid{Embed}\;\Varid{f}\mathord{@}(\Conid{Var}\;\anonymous )){}\<[48]%
\>[48]{}\mathrel{=}\Conid{Just}\;\Varid{f}{}\<[E]%
\\
\>[B]{}\Varid{isSimple}\;(\Conid{Map}\;\Varid{name}\;\Varid{f})\mid \Conid{Just}\;\Varid{f'}\leftarrow \Varid{isSimple}\;\Varid{f}{}\<[48]%
\>[48]{}\mathrel{=}\Conid{Just}\;(\Conid{Const}\;\Varid{name}\mathbin{`\Conid{Apply}`}\Varid{f'}){}\<[E]%
\\
\>[B]{}\Varid{isSimple}\;\anonymous {}\<[48]%
\>[48]{}\mathrel{=}\Conid{Nothing}{}\<[E]%
\ColumnHook
\end{pboxed}
\)\par\vskip\belowdisplayskip\noindent\endgroup\resethooks
\caption{module \ensuremath{\Conid{Generate}}, generation and simplification of free theorems}\label{fig:generate}
\end{figure*}

\begin{figure*}
\begingroup\par\vskip\abovedisplayskip\noindent\advance\leftskip\mathindent\(
\begin{pboxed}\SaveRestoreHook
\column{B}{@{}>{\hspre}l<{\hspost}@{}}%
\column{3}{@{}>{\hspre}l<{\hspost}@{}}%
\column{10}{@{}>{\hspre}l<{\hspost}@{}}%
\column{14}{@{}>{\hspre}c<{\hspost}@{}}%
\column{14E}{@{}l@{}}%
\column{17}{@{}>{\hspre}l<{\hspost}@{}}%
\column{21}{@{}>{\hspre}l<{\hspost}@{}}%
\column{28}{@{}>{\hspre}l<{\hspost}@{}}%
\column{34}{@{}>{\hspre}l<{\hspost}@{}}%
\column{36}{@{}>{\hspre}l<{\hspost}@{}}%
\column{56}{@{}>{\hspre}l<{\hspost}@{}}%
\column{E}{@{}>{\hspre}l<{\hspost}@{}}%
\>[B]{}\mathbf{module}\;\Conid{Syntax}\;(\Conid{Type}\;(\mathinner{\ldotp\ldotp}),\Conid{Func}\;(\mathinner{\ldotp\ldotp}),\Conid{Term}\;(\mathinner{\ldotp\ldotp}),\Varid{etaReduce})\;\mathbf{where}{}\<[E]%
\\[0.7ex]
\>[B]{}\mathbf{import}\;\Conid{\Conid{Control}.\Conid{Monad}.State}{}\<[E]%
\\[0.7ex]
\>[B]{}\mathbf{data}\;\Conid{Type}\mathrel{=}\Conid{Alpha}\mid \Conid{Bool}\mid \Conid{Int}\mid \Conid{List}\;\Conid{Type}\mid \Conid{Maybe}\;\Conid{Type}\mid \Conid{Type}\mathbin{`\Conid{To}`}\Conid{Type}{}\<[E]%
\\[0.7ex]
\>[B]{}\mathbf{data}\;\Conid{Func}\;\Varid{v}{}\<[14]%
\>[14]{}\mathrel{=}{}\<[14E]%
\>[17]{}\Conid{Id}\mid \Conid{Map}\;\Conid{String}\;(\Conid{Func}\;\Varid{v})\mid \Conid{Lambda}\;(\Varid{v}\to \Conid{Func}\;\Varid{v})\mid \Conid{Func}\;\Varid{v}\mathbin{`\Conid{Comp}`}\Conid{Func}\;\Varid{v}{}\<[E]%
\\
\>[14]{}\mid {}\<[14E]%
\>[17]{}\Conid{Embed}\;\Varid{v}{}\<[E]%
\\[0.7ex]
\>[B]{}\mathbf{data}\;\Conid{Term}\;\Varid{v}\mathrel{=}\Conid{Const}\;\Conid{String}\mid \Conid{Var}\;\Varid{v}\mid \Conid{Term}\;\Varid{v}\mathbin{`\Conid{Apply}`}\Conid{Term}\;\Varid{v}\mid \Conid{Lambda'}\;(\Varid{v}\to \Conid{Term}\;\Varid{v}){}\<[E]%
\\[0.7ex]
\>[B]{}\Varid{etaReduce}\mathbin{::}\Conid{Term}\;\Conid{String}\to \Conid{Term}\;\Conid{String}{}\<[E]%
\\
\>[B]{}\Varid{etaReduce}\;\Varid{t}\mathrel{=}\Varid{evalState}\;(\Varid{go}\;\Varid{t})\;[\mskip1.5mu \text{\ttfamily '*'}\mathbin{:}\Varid{show}\;\Varid{n}\mid \Varid{n}\leftarrow [\mskip1.5mu \mathrm{1}\mathinner{\ldotp\ldotp}\mskip1.5mu]\mskip1.5mu]{}\<[E]%
\\
\>[B]{}\hsindent{3}{}\<[3]%
\>[3]{}\mathbf{where}\;{}\<[10]%
\>[10]{}\Varid{go}\;\Varid{t}\mathord{@}(\Conid{Const}\;\anonymous ){}\<[28]%
\>[28]{}\mathrel{=}\Varid{return}\;\Varid{t}{}\<[E]%
\\
\>[10]{}\Varid{go}\;\Varid{t}\mathord{@}(\Conid{Var}\;\anonymous ){}\<[28]%
\>[28]{}\mathrel{=}\Varid{return}\;\Varid{t}{}\<[E]%
\\
\>[10]{}\Varid{go}\;(\Varid{f}\mathbin{`\Conid{Apply}`}\Varid{t}){}\<[28]%
\>[28]{}\mathrel{=}\Varid{liftM2}\;\Conid{Apply}\;(\Varid{go}\;\Varid{f})\;(\Varid{go}\;\Varid{t}){}\<[E]%
\\
\>[10]{}\Varid{go}\;(\Conid{Lambda'}\;\Varid{f}){}\<[28]%
\>[28]{}\mathrel{=}\mathbf{do}\;{}\<[34]%
\>[34]{}\Varid{v}\mathbin{:}\Varid{vs}\leftarrow \Varid{get}{}\<[E]%
\\
\>[34]{}\Varid{put}\;\Varid{vs}{}\<[E]%
\\
\>[34]{}\Varid{body}\leftarrow \Varid{go}\;(\Varid{f}\;\Varid{v}){}\<[E]%
\\
\>[34]{}\mathbf{case}\;\Varid{body}\;\mathbf{of}{}\<[E]%
\\
\>[34]{}\hsindent{2}{}\<[36]%
\>[36]{}\Varid{t}\mathbin{`\Conid{Apply}`}\Conid{Var}\;\Varid{v'}\mid {}\<[56]%
\>[56]{}\Varid{not}\;(\Varid{freeVar}\;\Varid{v}\;\Varid{t})\mathop{\,\&\&\,}\Varid{v'}\mathrel{==}\Varid{v}{}\<[E]%
\\
\>[56]{}\to \Varid{return}\;\Varid{t}{}\<[E]%
\\
\>[34]{}\hsindent{2}{}\<[36]%
\>[36]{}\anonymous \to \Varid{return}\;(\Conid{Lambda'}\;(\lambda \Varid{v}\to \Varid{evalState}\;(\Varid{go}\;(\Varid{f}\;\Varid{v}))\;\Varid{vs})){}\<[E]%
\\[0.7ex]
\>[10]{}\Varid{freeVar}\;\anonymous \;{}\<[21]%
\>[21]{}(\Conid{Const}\;\anonymous ){}\<[36]%
\>[36]{}\mathrel{=}\Conid{False}{}\<[E]%
\\
\>[10]{}\Varid{freeVar}\;\Varid{v}\;{}\<[21]%
\>[21]{}(\Conid{Var}\;\Varid{v'}){}\<[36]%
\>[36]{}\mathrel{=}\Varid{v'}\mathrel{==}\Varid{v}{}\<[E]%
\\
\>[10]{}\Varid{freeVar}\;\Varid{v}\;{}\<[21]%
\>[21]{}(\Varid{f}\mathbin{`\Conid{Apply}`}\Varid{t}){}\<[36]%
\>[36]{}\mathrel{=}\Varid{freeVar}\;\Varid{v}\;\Varid{f}\mathop{\,||\,}\Varid{freeVar}\;\Varid{v}\;\Varid{t}{}\<[E]%
\\
\>[10]{}\Varid{freeVar}\;\Varid{v}\;{}\<[21]%
\>[21]{}(\Conid{Lambda'}\;\Varid{f}){}\<[36]%
\>[36]{}\mathrel{=}\Varid{freeVar}\;\Varid{v}\;(\Varid{f}\;\text{\ttfamily \char34 \char34}){}\<[E]%
\ColumnHook
\end{pboxed}
\)\par\vskip\belowdisplayskip\noindent\endgroup\resethooks
\caption{module \ensuremath{\Conid{Syntax}}, datatypes for types and different forms of (higher-order abstract syntax) terms, and eta-reduction}\label{fig:syntax}
\end{figure*}

\begin{figure*}
\begingroup\par\vskip\abovedisplayskip\noindent\advance\leftskip\mathindent\(
\begin{pboxed}\SaveRestoreHook
\column{B}{@{}>{\hspre}l<{\hspost}@{}}%
\column{3}{@{}>{\hspre}l<{\hspost}@{}}%
\column{8}{@{}>{\hspre}l<{\hspost}@{}}%
\column{E}{@{}>{\hspre}l<{\hspost}@{}}%
\>[B]{}\mathbf{module}\;\Conid{Main}\;(\Varid{main})\;\mathbf{where}{}\<[E]%
\\[0.7ex]
\>[B]{}\mathbf{import}\;\Conid{Syntax}\;(\Conid{Type}\;(\mathinner{\ldotp\ldotp}),\Conid{Func}\;(\mathinner{\ldotp\ldotp}),\Conid{Term}\;(\mathinner{\ldotp\ldotp}),\Varid{etaReduce}){}\<[E]%
\\
\>[B]{}\mathbf{import}\;\Conid{Parser}\;(\Varid{parse}){}\<[E]%
\\
\>[B]{}\mathbf{import}\;\Conid{Generate}\;(\Varid{mono},\Varid{apply}){}\<[E]%
\\
\>[B]{}\mathbf{import}\;\Conid{Show}\;(){}\<[E]%
\\[0.7ex]
\>[B]{}\mathbf{import}\;\Conid{\Conid{System}.IO}{}\<[E]%
\\[0.7ex]
\>[B]{}\Varid{sigma}\mathbin{::}\Conid{Type}{}\<[E]%
\\
\>[B]{}\Varid{sigma}\mathrel{=}(\Conid{Alpha}\mathbin{`\Conid{To}`}\Conid{Bool})\mathbin{`\Conid{To}`}((\Conid{Bool}\mathbin{`\Conid{To}`}\Conid{Alpha})\mathbin{`\Conid{To}`}(\Conid{List}\;\Conid{Alpha}\mathbin{`\Conid{To}`}\Conid{Alpha})){}\<[E]%
\\[0.7ex]
\>[B]{}\Varid{main}\mathbin{::}\Conid{IO}\;(){}\<[E]%
\\
\>[B]{}\Varid{main}\mathrel{=}\mathbf{do}{}\<[E]%
\\
\>[B]{}\hsindent{3}{}\<[3]%
\>[3]{}\Varid{hSetBuffering}\;\Varid{stdout}\;\Conid{NoBuffering}{}\<[E]%
\\
\>[B]{}\hsindent{3}{}\<[3]%
\>[3]{}\Varid{putStr}\mathbin{\$}\text{\ttfamily \char34 function~type~(or~Enter~for~default):~\char34}{}\<[E]%
\\
\>[B]{}\hsindent{3}{}\<[3]%
\>[3]{}\Varid{sigma}\leftarrow \Varid{getLine}\bind \lambda \Varid{s}\to \Varid{return}\mathbin{\$}\mathbf{if}\;\Varid{s}\mathrel{==}\text{\ttfamily \char34 \char34}\;\mathbf{then}\;\Varid{sigma}\;\mathbf{else}\;\Varid{parse}\;\Varid{s}{}\<[E]%
\\
\>[B]{}\hsindent{3}{}\<[3]%
\>[3]{}\Varid{putStrLn}\;\text{\ttfamily \char34 \char34}{}\<[E]%
\\
\>[B]{}\hsindent{3}{}\<[3]%
\>[3]{}\Varid{putStrLn}\mathbin{\$}\text{\ttfamily \char34 f~::~\char34}\plus \Varid{show}\;\Varid{sigma}{}\<[E]%
\\
\>[B]{}\hsindent{3}{}\<[3]%
\>[3]{}\Varid{putStrLn}\mathbin{\$}\Varid{replicate}\;\mathrm{66}\;\text{\ttfamily '-'}{}\<[E]%
\\
\>[B]{}\hsindent{3}{}\<[3]%
\>[3]{}\Varid{putStrLn}\mathbin{\$}\text{\ttfamily \char34 e~=~\char34}\plus \Varid{show}\;(\Varid{mono}\;\Varid{sigma}\;(\Conid{Embed}\;\text{\ttfamily \char34 pre\char34})\;(\Conid{Embed}\;\text{\ttfamily \char34 post\char34}))\plus \text{\ttfamily \char34 ~f\char34}{}\<[E]%
\\
\>[B]{}\hsindent{3}{}\<[3]%
\>[3]{}\Varid{putStrLn}\mathbin{\$}\Varid{replicate}\;\mathrm{66}\;\text{\ttfamily '-'}{}\<[E]%
\\
\>[B]{}\hsindent{3}{}\<[3]%
\>[3]{}\mathbf{let}\;{}\<[8]%
\>[8]{}\Varid{lhs}\mathrel{=}\Varid{mono}\;\Varid{sigma}\;\Conid{Id}\;(\Conid{Embed}\;(\Conid{Var}\;\text{\ttfamily \char34 g\char34}))\mathbin{`\Varid{apply}`}\Conid{Const}\;\text{\ttfamily \char34 f\char34}{}\<[E]%
\\
\>[8]{}\Varid{rhs}\mathrel{=}\Varid{mono}\;\Varid{sigma}\;(\Conid{Embed}\;(\Conid{Var}\;\text{\ttfamily \char34 g\char34}))\;\Conid{Id}\mathbin{`\Varid{apply}`}\Conid{Const}\;\text{\ttfamily \char34 f\char34}{}\<[E]%
\\
\>[B]{}\hsindent{3}{}\<[3]%
\>[3]{}\Varid{putStrLn}\mathbin{\$}\text{\ttfamily \char34 free~theorem:\char34}{}\<[E]%
\\
\>[B]{}\hsindent{3}{}\<[3]%
\>[3]{}\Varid{putStrLn}\mathbin{\$}\text{\ttfamily \char34 ~\char34}\plus \Varid{show}\;\Varid{lhs}{}\<[E]%
\\
\>[B]{}\hsindent{3}{}\<[3]%
\>[3]{}\Varid{putStrLn}\;\text{\ttfamily \char34 ~~=\char34}{}\<[E]%
\\
\>[B]{}\hsindent{3}{}\<[3]%
\>[3]{}\Varid{putStrLn}\mathbin{\$}\text{\ttfamily \char34 ~\char34}\plus \Varid{show}\;\Varid{rhs}{}\<[E]%
\\
\>[B]{}\hsindent{3}{}\<[3]%
\>[3]{}\Varid{putStrLn}\mathbin{\$}\Varid{replicate}\;\mathrm{66}\;\text{\ttfamily '-'}{}\<[E]%
\\
\>[B]{}\hsindent{3}{}\<[3]%
\>[3]{}\Varid{putStrLn}\mathbin{\$}\text{\ttfamily \char34 free~theorem,~eta-reduced:\char34}{}\<[E]%
\\
\>[B]{}\hsindent{3}{}\<[3]%
\>[3]{}\Varid{putStrLn}\mathbin{\$}\text{\ttfamily \char34 ~\char34}\plus \Varid{show}\;(\Varid{etaReduce}\;\Varid{lhs}){}\<[E]%
\\
\>[B]{}\hsindent{3}{}\<[3]%
\>[3]{}\Varid{putStrLn}\;\text{\ttfamily \char34 ~~=\char34}{}\<[E]%
\\
\>[B]{}\hsindent{3}{}\<[3]%
\>[3]{}\Varid{putStrLn}\mathbin{\$}\text{\ttfamily \char34 ~\char34}\plus \Varid{show}\;(\Varid{etaReduce}\;\Varid{rhs}){}\<[E]%
\\
\>[B]{}\hsindent{3}{}\<[3]%
\>[3]{}\Varid{putStrLn}\;\text{\ttfamily \char34 \char34}{}\<[E]%
\ColumnHook
\end{pboxed}
\)\par\vskip\belowdisplayskip\noindent\endgroup\resethooks
\caption{module \ensuremath{\Conid{Main}}, putting the generator together with input and output}\label{fig:main}
\end{figure*}

\begin{figure*}
\begin{tabbing}\ttfamily
~\char42{}Main\char62{}~main\\
\ttfamily ~function~type~\char40{}or~Enter~for~default\char41{}\char58{}~\\
\ttfamily ~\\
\ttfamily ~f~\char58{}\char58{}~\char40{}alpha~\char45{}\char62{}~Bool\char41{}~\char45{}\char62{}~\char40{}Bool~\char45{}\char62{}~alpha\char41{}~\char45{}\char62{}~\char91{}alpha\char93{}~\char45{}\char62{}~alpha\\
\ttfamily ~\char45{}\char45{}\char45{}\char45{}\char45{}\char45{}\char45{}\char45{}\char45{}\char45{}\char45{}\char45{}\char45{}\char45{}\char45{}\char45{}\char45{}\char45{}\char45{}\char45{}\char45{}\char45{}\char45{}\char45{}\char45{}\char45{}\char45{}\char45{}\char45{}\char45{}\char45{}\char45{}\char45{}\char45{}\char45{}\char45{}\char45{}\char45{}\char45{}\char45{}\char45{}\char45{}\char45{}\char45{}\char45{}\char45{}\char45{}\char45{}\char45{}\char45{}\char45{}\char45{}\char45{}\char45{}\char45{}\char45{}\char45{}\char45{}\char45{}\char45{}\char45{}\char45{}\char45{}\char45{}\char45{}\char45{}\\
\ttfamily ~e~\char61{}~\char40{}\char92{}h1~\char45{}\char62{}~\char40{}\char92{}h2~\char45{}\char62{}~\char40{}\char92{}h3~\char45{}\char62{}~post~\char46{}~h3~\char46{}~map~pre\char41{}\\
\ttfamily ~~~~~~~~~~~~~~~~~~~~~\char46{}~h2~\char46{}~\char40{}\char92{}h4~\char45{}\char62{}~pre~\char46{}~h4~\char46{}~id\char41{}\char41{}\\
\ttfamily ~~~~~~~~~~~~~\char46{}~h1~\char46{}~\char40{}\char92{}h5~\char45{}\char62{}~id~\char46{}~h5~\char46{}~post\char41{}\char41{}~f\\
\ttfamily ~\char45{}\char45{}\char45{}\char45{}\char45{}\char45{}\char45{}\char45{}\char45{}\char45{}\char45{}\char45{}\char45{}\char45{}\char45{}\char45{}\char45{}\char45{}\char45{}\char45{}\char45{}\char45{}\char45{}\char45{}\char45{}\char45{}\char45{}\char45{}\char45{}\char45{}\char45{}\char45{}\char45{}\char45{}\char45{}\char45{}\char45{}\char45{}\char45{}\char45{}\char45{}\char45{}\char45{}\char45{}\char45{}\char45{}\char45{}\char45{}\char45{}\char45{}\char45{}\char45{}\char45{}\char45{}\char45{}\char45{}\char45{}\char45{}\char45{}\char45{}\char45{}\char45{}\char45{}\char45{}\char45{}\char45{}\\
\ttfamily ~free~theorem\char58{}\\
\ttfamily ~~\char92{}x1~x2~x3~\char45{}\char62{}~g~\char40{}f~\char40{}\char92{}x4~\char45{}\char62{}~x1~\char40{}g~x4\char41{}\char41{}~\char40{}\char92{}x5~\char45{}\char62{}~x2~x5\char41{}~x3\char41{}\\
\ttfamily ~~~\char61{}\\
\ttfamily ~~\char92{}x1~x2~x3~\char45{}\char62{}~f~\char40{}\char92{}x4~\char45{}\char62{}~x1~x4\char41{}~\char40{}\char92{}x5~\char45{}\char62{}~g~\char40{}x2~x5\char41{}\char41{}~\char40{}map~g~x3\char41{}\\
\ttfamily ~\char45{}\char45{}\char45{}\char45{}\char45{}\char45{}\char45{}\char45{}\char45{}\char45{}\char45{}\char45{}\char45{}\char45{}\char45{}\char45{}\char45{}\char45{}\char45{}\char45{}\char45{}\char45{}\char45{}\char45{}\char45{}\char45{}\char45{}\char45{}\char45{}\char45{}\char45{}\char45{}\char45{}\char45{}\char45{}\char45{}\char45{}\char45{}\char45{}\char45{}\char45{}\char45{}\char45{}\char45{}\char45{}\char45{}\char45{}\char45{}\char45{}\char45{}\char45{}\char45{}\char45{}\char45{}\char45{}\char45{}\char45{}\char45{}\char45{}\char45{}\char45{}\char45{}\char45{}\char45{}\char45{}\char45{}\\
\ttfamily ~free~theorem\char44{}~eta\char45{}reduced\char58{}\\
\ttfamily ~~\char92{}x1~x2~x3~\char45{}\char62{}~g~\char40{}f~\char40{}\char92{}x4~\char45{}\char62{}~x1~\char40{}g~x4\char41{}\char41{}~x2~x3\char41{}\\
\ttfamily ~~~\char61{}\\
\ttfamily ~~\char92{}x1~x2~x3~\char45{}\char62{}~f~x1~\char40{}\char92{}x4~\char45{}\char62{}~g~\char40{}x2~x4\char41{}\char41{}~\char40{}map~g~x3\char41{}\\
\ttfamily ~\\
\ttfamily ~\char42{}Main\char62{}~main\\
\ttfamily ~function~type~\char40{}or~Enter~for~default\char41{}\char58{}~\char40{}a~\char45{}\char62{}~a~\char45{}\char62{}~Bool\char41{}~\char45{}\char62{}~\char91{}a\char93{}~\char45{}\char62{}~\char91{}a\char93{}\\
\ttfamily ~\\
\ttfamily ~f~\char58{}\char58{}~\char40{}alpha~\char45{}\char62{}~alpha~\char45{}\char62{}~Bool\char41{}~\char45{}\char62{}~\char91{}alpha\char93{}~\char45{}\char62{}~\char91{}alpha\char93{}\\
\ttfamily ~\char45{}\char45{}\char45{}\char45{}\char45{}\char45{}\char45{}\char45{}\char45{}\char45{}\char45{}\char45{}\char45{}\char45{}\char45{}\char45{}\char45{}\char45{}\char45{}\char45{}\char45{}\char45{}\char45{}\char45{}\char45{}\char45{}\char45{}\char45{}\char45{}\char45{}\char45{}\char45{}\char45{}\char45{}\char45{}\char45{}\char45{}\char45{}\char45{}\char45{}\char45{}\char45{}\char45{}\char45{}\char45{}\char45{}\char45{}\char45{}\char45{}\char45{}\char45{}\char45{}\char45{}\char45{}\char45{}\char45{}\char45{}\char45{}\char45{}\char45{}\char45{}\char45{}\char45{}\char45{}\char45{}\char45{}\\
\ttfamily ~e~\char61{}~\char40{}\char92{}h1~\char45{}\char62{}~\char40{}\char92{}h2~\char45{}\char62{}~map~post~\char46{}~h2~\char46{}~map~pre\char41{}\\
\ttfamily ~~~~~~~~~~~~~\char46{}~h1~\char46{}~\char40{}\char92{}h3~\char45{}\char62{}~\char40{}\char92{}h4~\char45{}\char62{}~id~\char46{}~h4~\char46{}~post\char41{}~\char46{}~h3~\char46{}~post\char41{}\char41{}~f\\
\ttfamily ~\char45{}\char45{}\char45{}\char45{}\char45{}\char45{}\char45{}\char45{}\char45{}\char45{}\char45{}\char45{}\char45{}\char45{}\char45{}\char45{}\char45{}\char45{}\char45{}\char45{}\char45{}\char45{}\char45{}\char45{}\char45{}\char45{}\char45{}\char45{}\char45{}\char45{}\char45{}\char45{}\char45{}\char45{}\char45{}\char45{}\char45{}\char45{}\char45{}\char45{}\char45{}\char45{}\char45{}\char45{}\char45{}\char45{}\char45{}\char45{}\char45{}\char45{}\char45{}\char45{}\char45{}\char45{}\char45{}\char45{}\char45{}\char45{}\char45{}\char45{}\char45{}\char45{}\char45{}\char45{}\char45{}\char45{}\\
\ttfamily ~free~theorem\char58{}\\
\ttfamily ~~\char92{}x1~x2~\char45{}\char62{}~map~g~\char40{}f~\char40{}\char92{}x3~x4~\char45{}\char62{}~x1~\char40{}g~x3\char41{}~\char40{}g~x4\char41{}\char41{}~x2\char41{}\\
\ttfamily ~~~\char61{}\\
\ttfamily ~~\char92{}x1~x2~\char45{}\char62{}~f~\char40{}\char92{}x3~x4~\char45{}\char62{}~x1~x3~x4\char41{}~\char40{}map~g~x2\char41{}\\
\ttfamily ~\char45{}\char45{}\char45{}\char45{}\char45{}\char45{}\char45{}\char45{}\char45{}\char45{}\char45{}\char45{}\char45{}\char45{}\char45{}\char45{}\char45{}\char45{}\char45{}\char45{}\char45{}\char45{}\char45{}\char45{}\char45{}\char45{}\char45{}\char45{}\char45{}\char45{}\char45{}\char45{}\char45{}\char45{}\char45{}\char45{}\char45{}\char45{}\char45{}\char45{}\char45{}\char45{}\char45{}\char45{}\char45{}\char45{}\char45{}\char45{}\char45{}\char45{}\char45{}\char45{}\char45{}\char45{}\char45{}\char45{}\char45{}\char45{}\char45{}\char45{}\char45{}\char45{}\char45{}\char45{}\char45{}\char45{}\\
\ttfamily ~free~theorem\char44{}~eta\char45{}reduced\char58{}\\
\ttfamily ~~\char92{}x1~x2~\char45{}\char62{}~map~g~\char40{}f~\char40{}\char92{}x3~x4~\char45{}\char62{}~x1~\char40{}g~x3\char41{}~\char40{}g~x4\char41{}\char41{}~x2\char41{}\\
\ttfamily ~~~\char61{}\\
\ttfamily ~~\char92{}x1~x2~\char45{}\char62{}~f~x1~\char40{}map~g~x2\char41{}
\end{tabbing}
  \caption{An example session}
  \label{fig:session}
\end{figure*}

\end{document}